\documentclass[12pt]{article}
\usepackage[a4paper,textwidth=126mm,textheight=195mm]{geometry}
\usepackage{amsmath,amssymb,bm}
\usepackage{amsthm}
\usepackage{graphicx}
\usepackage{booktabs}
\usepackage{url}
\usepackage{algorithm}
\usepackage{algorithmic}

\newtheorem{theorem}{Theorem}[section]
\newtheorem{definition}[theorem]{Definition}
\newtheorem{lemma}[theorem]{Lemma}
\newtheorem{corollary}[theorem]{Corollary}
\newtheorem{proposition}[theorem]{Proposition}
\newtheorem{conjecture}[theorem]{Conjecture}
\newtheorem{remark}[theorem]{Remark}
\newtheorem{example}[theorem]{Example}

\title{Structural Properties of DomiRank Centrality:\\
Symmetry Axioms, an Entropy Bound, and Spectral Interpolation}
\author{Yingying Zhang\textsuperscript{1}, Chengye Zhao\textsuperscript{1,*}\\
{\small 1. College of Sciences, China Jiliang University, Hangzhou 310018, China}\\
{\small * Corresponding author. E-mail: cyzhao@cjlu.edu.cn}\\
{\small ORCID: 0009-0001-9073-2525 (Y. Zhang); 0000-0002-1149-4419 (C. Zhao)}}

\begin{document}
\maketitle

\begin{abstract}
DomiRank is a node-importance algorithm for unweighted networks, defined by a dynamical model governed by a competition-strength parameter, a dominance threshold, and a decay rate. We establish three structural results. First, on symmetry. Automorphism-related vertices share the same DomiRank value, and a graph whose values are pairwise distinct must be asymmetric. The orbit count bounds the number of distinct values, and regular and vertex-transitive graphs admit closed forms. Second, on entropy. The maximum of the DomiRank entropy over connected graphs, $\log_2 N$, is attained if and only if the graph is regular. The initial entropy derivative is the degree--log-degree edge covariance, making the decrease rigorous near vanishing competition for structurally disassortative graphs, and an assortative counterexample shows it does not extend to all graphs; a sufficient-condition theorem and a refined conjecture cover the rest. Third, on spectrum. The normalized distribution is $\sigma$-invariant if and only if the degree vector is an adjacency-matrix eigenvector. The parameter $\sigma$ interpolates continuously between degree and smallest-eigenvector centrality within the resolvent (Bonacich-type) family, its sensitivity being a second-order suppression operator on the degree vector. Graph transformations are also characterized, including the $m$-product. Experiments on four real networks and synthetic families up to $10^4$ vertices confirm the interpolation and exhibit monotone entropy decrease, continuous ranking interpolation, decaying positive fraction and ranking reversals; where the positive fraction has decayed far from one, the limit describes the solution vector rather than the induced ranking. DomiRank is thus a tunable complement to principal-eigenvector centrality. A companion study quantifies how much of its task-level edge over a degree baseline survives degree correction.
\end{abstract}


\section{Introduction}

In the analysis of complex networks, node-importance evaluation is a fundamental task for revealing structural features and mining core functional units. Related algorithms have been widely applied to social-network analysis \cite{p1,p2}, information-propagation prediction \cite{p3}, network modeling \cite{p4}, and critical-infrastructure protection \cite{p5,p6}. Among traditional methods, degree centrality \cite{p7} counts the number of direct neighbors and is simple and intuitive, but it ignores the overall structure and indirect connections; eigenvector centrality \cite{p8} emphasizes that a node's importance depends on the quality of its neighbors, yet suffers from low computational efficiency on large networks; PageRank \cite{p9} improves the iterative logic of eigenvector centrality by introducing a damping factor and has become a classic tool for Web ranking and citation networks, but its essence still relies on link-weight propagation and it cannot accurately capture the suppression effect and dominance capability of nodes over their neighbors. Moreover, the structural-symmetry information encoded in graph automorphisms --- and, more generally, the question of which vertices a centrality index is able to distinguish --- is rarely analyzed explicitly when a new index is proposed; the benchmark framework of \cite{p10} makes such an analysis systematic for existing indices.

DomiRank \cite{p11} was proposed to address these limitations. Based on a dynamical-system model with a competition-strength parameter and a dominance threshold, it builds an evaluation framework that quantifies node dominance, offering a new perspective on node-importance analysis. The graph automorphism group \cite{p12}, the classical algebraic object describing structural symmetry, reveals the equivalence mappings among vertices; its combination with node-importance evaluation provides important theoretical support for the soundness and robustness of the algorithm. At the same time, entropy \cite{p13}, as an effective tool for measuring the uniformity of a distribution, reflects the degree of differentiation of node importance and provides a quantitative basis for understanding the influence of algorithm parameters on the evaluation results.

These three aspects are complementary: graph automorphism characterizes ``which vertices are structurally equivalent'' (symmetry constraints), DomiRank entropy measures ``the degree of differentiation of the importance distribution'' (distributional uniformity), and graph transformations reveal ``the stability laws of the evaluation results under structural interventions'' (robustness). They examine the same theoretical object, DomiRank, from the three perspectives of equivalence classes, measures, and structural dynamics, jointly forming a systematic theoretical-analysis framework for this centrality.

The difficulty is that DomiRank has so far been used as a practical ranking algorithm, and its three defining features have been treated as implementation details rather than as objects of theory. The steady state of \eqref{eq:dynamic} depends on the competition parameter, so the vector it converges to is not one of the classical eigenvector-based rankings and no formula for those applies to it; the entropy it induces is not one of the standard graph entropies, so the usual extremal arguments do not transfer; and the graph transformations that practitioners apply to real networks (smoothing, partitioning, edge deletion, blow-up) have no known effect on the values. In what follows we answer the three questions this raises: what the score can and cannot distinguish, how its dispersion behaves as $\sigma$ grows, and which of its algebraic features survive a structural intervention.

The main contributions of this paper are as follows:
\begin{enumerate}
    \item \textbf{A symmetry axiom with an exact combinatorial strength.} We prove that DomiRank is invariant under the automorphism group (Theorem~\ref{thm:automorphism}) and that the number of distinct DomiRank values is bounded by the number of orbits (Corollary~\ref{cor:orbit_count}), with equality holding on every vertex-transitive graph (Theorem~\ref{thm:two_orbits}) and pairwise distinct values forcing asymmetry (Corollary~\ref{cor:asymmetric}); we exhibit the smallest-order instance in which structurally inequivalent vertices receive equal values, which shows the bound to be strict (Section~\ref{sec:automorphism}).

    \item \textbf{An entropy bound with a closed-form initial derivative.} We prove that the maximum DomiRank entropy $\log_2 N$ is attained if and only if the graph is regular (Theorem~\ref{thm:entropy_max}), compute the initial entropy derivative in closed form as the degree--log-degree edge covariance (Proposition~\ref{prop:initial_entropy}), and show through an assortative counterexample that unconditional monotonicity fails; the surviving statement is formulated as Conjecture~\ref{conj:entropy} and supported by exhaustive enumeration.

    \item \textbf{The spectral role of $\sigma$, and where it fails.} Within the resolvent (Bonacich-type) family, $\sigma$ moves the ranking continuously from degree centrality to smallest-eigenvector centrality (Theorem~\ref{thm:sigma_zero}, Proposition~\ref{prop:limit}), with a closed-form sensitivity formula in Proposition~\ref{prop:sensitivity}, and we give the exact characterization of its degeneracy: the normalized distribution is $\sigma$-invariant if and only if the graph is harmonic (Proposition~\ref{prop:invariance}), with the interpolation mechanism and its three axes of difference from the existing theory of parameter-dependent centralities made explicit in Remark~\ref{rem:interpolation}.

    \item \textbf{Behaviour under structural transformations.} We derive the response of DomiRank to vertex similarity, vertex partitions, edge swaps and $m$-products (Theorems~\ref{thm:similar_neighbor}, \ref{thm:partition}--\ref{thm:m_product}), and we treat edge deletion separately (Theorem~\ref{thm:edge_deletion}), where an identity survives under mild conditions but the stronger exact scaling fails, with the smallest counterexample recorded; the $m$-product case gives the explicit relation between the values on $G^{(m)}$ and those on $G$ at parameter $\sigma m$.
\end{enumerate}

Positioning within the centrality literature. It is worth making explicit where the present contribution sits, since each of its three pillars stands on an established line of work. The symmetry pillar answers a question posed long ago by Sabidussi \cite{p14} and formalized as a benchmark by Bao and Zhang \cite{p10}: a centrality index is sound only if structurally equivalent vertices receive equal values, and its discriminating power is then read off from the extent to which inequivalent vertices receive different ones. We adopt that benchmark rather than propose a new one, and we show that DomiRank passes it (Theorem~\ref{thm:automorphism}) while adding a parameter-dependent ranking evolution that parameter-free indices such as degree or eigenvector centrality cannot express (Section~\ref{sec:sigma}). The entropy pillar places DomiRank among the graph-entropy measures surveyed by Dehmer and Mowshowitz \cite{p13}, and contributes an exact maximum, a closed-form initial derivative with a signed combinatorial interpretation, and a documented counterexample to unconditional monotonicity. The spectral pillar places DomiRank in the classical resolvent, or Bonacich-type, family \cite{p15,p16,p17}, and shows that its distinctive behaviour, namely non-positive values together with a smallest-eigenvector limit, arises from the negative-parameter regime and from the choice $\boldsymbol{\beta} = \boldsymbol{k}$, rather than from the mere presence of a parameter (Remark~\ref{rem:interpolation}). This last point requires care, because the parameter-dependent behaviour of resolvent centralities has been studied systematically: Benzi and Klymko \cite{p18} analysed the limiting behaviour of Katz-type and subgraph-type centrality measures as the algebraic parameter varies, showing explicitly how the family interpolates between degree-like and eigenvector-like regimes. We do not claim the interpolation phenomenon itself as new; Remark~\ref{rem:interpolation} states the three axes along which the present result differs, the decisive one being that our endpoint is the \emph{smallest} rather than the principal eigenvector, which is reachable only because the competitive parameter is allowed to be negative. Two of the remaining ingredients are already known separately in the literature: vertex distinction through subgraph centrality was resolved by Ballini and Deniskin \cite{p19}, and the bipartitivity information carried by the smallest eigenvector is classical \cite{p20,p21}. Our contribution is to show that a single physically interpretable parameter ties the three together, and to state exactly where that unification fails (harmonic graphs, Proposition~\ref{prop:invariance}).

In parallel with this theoretical study, the companion paper \cite{p22} exploits the tunability of $\sigma$ established here to construct DomiRank-based fusion strategies for link prediction, demonstrating the practical value of the properties developed in this work; it is available as arXiv:2609.09658. A second companion, on null-model tests of centrality specificity, is cited in Section~\ref{sec:relation} as unpublished.

\paragraph{Scope.} The present paper is concerned exclusively with \emph{single-layer} simple undirected graphs. The extension of DomiRank to multiplex networks by Zheng, Engsig, Tejedor and Moreno \cite{p23} combines layer-specific DomiRank values with a global contextualization step in which multiplex-wide DomiRank scores are used to combine the per-layer rankings. That construction is not reproduced here, and no claim of this paper is made for multiplex graphs: symmetry orbits, the entropy bound, and the graph-transformation theorems below all refer to a single adjacency matrix. The relationship is one of foundation rather than of extension --- the single-layer analysis developed here is the necessary starting point for a spectral account of the multiplex version, and the degeneracy result of Proposition~\ref{prop:invariance} in particular singles out the graphs on which the competitive parameter carries no information at all, a condition that a naive layer-wise application of DomiRank would silently violate.

\section{Preliminaries}

We treat a network as an unweighted graph. Let $N$ be the number of vertices and assume $A \in \mathbb{R}^{N \times N}$ is the adjacency matrix of the network, defined by
\[
A_{ij}=
\begin{cases}
1 & \text{if vertices } i \text{ and } j \text{ are adjacent},\\
0 & \text{otherwise.}
\end{cases}
\]

We define the DomiRank centrality (denoted by $\boldsymbol{\Gamma} \in \mathbb{R}^{N \times 1}$) as the steady-state solution of the dynamical-system model, governed by
\begin{align}\label{eq:dynamic}
	\frac{d\boldsymbol{\Gamma}(t)}{dt} = \alpha A\left(\theta \mathbf{1}_{N \times 1} - \boldsymbol{\Gamma}(t)\right) - \beta \boldsymbol{\Gamma}(t),
\end{align}
where $\alpha$ is the competition-strength parameter controlling the suppression effect among neighbors, $\theta$ is the dominance threshold (we take $\theta = 1$ throughout), and $\beta$ is the natural decay rate ensuring convergence.

For the $i$-th vertex, equation \eqref{eq:dynamic} can be written as
\begin{align}\label{eq:node}
	\frac{d\Gamma_i(t)}{dt} = \alpha \sum_{j \in N(i)} \left(\theta - \Gamma_j(t)\right) - \beta \Gamma_i(t),
\end{align}
where $N(i)$ denotes the set of vertices directly connected to vertex $i$.

Theorem~\ref{thm:analytical} below reduces this dynamical system to a single linear solve, and every result in this paper is a statement about the vector it returns; we therefore take it as the starting point rather than re-deriving it in each section.

\begin{theorem}\label{thm:analytical}
When system \eqref{eq:dynamic} reaches its steady state, the analytic solution of the DomiRank centrality $\boldsymbol{\Gamma} \in \mathbb{R}^{N \times 1}$ is
\begin{align}\label{eq:solution}
\boldsymbol{\Gamma} = \theta \sigma \left(\sigma A + I\right)^{-1} A \mathbf{1}_{N \times 1},
\end{align}
where $\sigma = \dfrac{\alpha}{\beta} \in \mathbb{R}^+$ and $\det(\sigma A + I) \neq 0$. Throughout, the eigenvalues of $A$ are ordered $\lambda_1 \ge \lambda_2 \ge \cdots \ge \lambda_N$, so that $\lambda_1$ is the largest (algebraically) and $\lambda_N$ the smallest eigenvalue; $\lambda_N < 0$ for every connected graph with $N \ge 2$. In particular, $\sigma \in \left(0, -\dfrac{1}{\lambda_N}\right)$, where $\lambda_N$ is the smallest eigenvalue of $A$, in which case the system converges.

For each vertex $i$,
\begin{align}\label{eq:node_solution}
\Gamma_i = \sigma \sum_{j \in N(i)} \left( \theta - \Gamma_j \right).
\end{align}
\end{theorem}

\begin{proof}
At steady state, $\dfrac{d\boldsymbol{\Gamma}(t)}{dt} = 0$, i.e.,
\[
\alpha A\left(\theta \mathbf{1}_{N \times 1} - \boldsymbol{\Gamma}\right) = \beta \boldsymbol{\Gamma}.
\]
Expanding,
\[
\alpha \theta A \mathbf{1}_{N \times 1} - \alpha A \boldsymbol{\Gamma} = \beta \boldsymbol{\Gamma}.
\]
Rearranging,
\[
\frac{\alpha}{\beta} \theta A \mathbf{1}_{N \times 1} = \left( \frac{\alpha}{\beta} A + I_{N} \right) \boldsymbol{\Gamma}.
\]
Letting $\sigma = \dfrac{\alpha}{\beta}$,
\[
\boldsymbol{\Gamma} = \sigma \theta \left( \sigma A + I_{N} \right)^{-1} A \mathbf{1}_{N \times 1}.
\]

From the steady-state form of equation \eqref{eq:node}, $\beta \Gamma_i = \alpha \sum_{j \in N(i)} (\theta - \Gamma_j)$, and with $\sigma = \alpha / \beta$ we obtain
\[
\Gamma_i = \sigma \sum_{j \in N(i)} \left( \theta - \Gamma_j \right).
\]
\end{proof}

\begin{remark}
To guarantee convergence, $\sigma A + I$ must be invertible and positive definite. By the spectral properties of $A$, when $\sigma \in (0, -1/\lambda_N)$ all eigenvalues of $\sigma A + I$ are positive, hence the matrix is invertible.
\end{remark}

\begin{example}\label{ex:first}
Consider the 5-vertex network of Figure~\ref{figure1}, whose adjacency matrix is
\[
A =
\begin{bmatrix}
0 & 1 & 1 & 0 & 0 \\
1 & 0 & 1 & 0 & 0 \\
1 & 1 & 0 & 1 & 1 \\
0 & 0 & 1 & 0 & 0 \\
0 & 0 & 1 & 0 & 0
\end{bmatrix}.
\]

The smallest negative eigenvalue of $A$ is $\lambda_{\min} \approx -1.814$, hence $\sigma \in \left(0, 0.551\right)$.

Taking $\sigma = 0.25$ and $\theta = 1$, equation \eqref{eq:solution} yields
\[
\boldsymbol{\Gamma} = \sigma \left(\sigma A + I\right)^{-1} A \mathbf{1} =
\begin{bmatrix}
0.2258 \\
0.2258 \\
0.8710 \\
0.0323 \\
0.0323
\end{bmatrix}.
\]
After normalization, $\begin{bmatrix} 0.1628 & 0.1628 & 0.6279 & 0.0233 & 0.0233 \end{bmatrix}^T$.
\end{example}

\begin{figure}[!h]
	\centering
	\includegraphics[width=\textwidth]{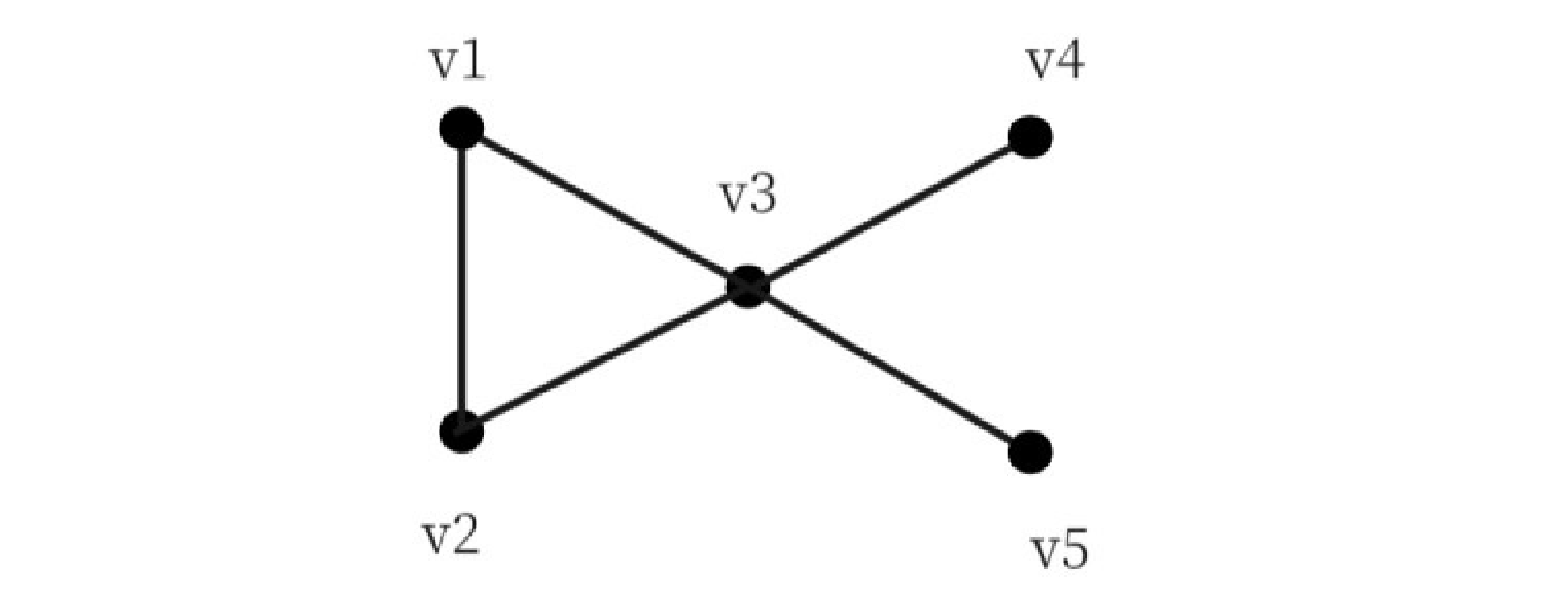}\\
	\caption{The 5-vertex graph of Example~\ref{ex:first}; vertices 1 and 2 form one orbit and vertices 4 and 5 another.}
	\textbf{Alt text:} A five-vertex undirected graph drawn with vertex 3 at the centre, joined to all four others; vertices 1 and 2 are joined to each other and each joined to vertex 3, and vertices 4 and 5 are leaves attached only to vertex 3. This is the graph of Example 2.1, with automorphism orbits {1,2}, {3} and {4,5}.
	\label{figure1}
\end{figure}

\section{DomiRank and Graph Automorphism}
\label{sec:automorphism}

The relation between graph automorphisms and centrality indices has a long history: Sabidussi \cite{p14} already used the property that structurally equivalent vertices should take equal values when discussing centrality indices in 1966; Bao and Zhang \cite{p10} further established ``same value on the same orbit, different values on different orbits'' as the benchmark framework for evaluating the discriminating power of centrality indices; for PageRank, Augeri \cite{p24} proved that vertices in the same block of the coarsest equitable partition necessarily share equal PageRank values. Recent work has also replaced the exact automorphism group by numerical surrogates, computing approximate structural symmetry from centrality-guided local iterations \cite{p25}; that line is algorithmically useful but yields no exact orbit partition, and the exact group is what the present analysis uses. Within this lineage, this section proves that automorphism-mapped vertices have equal DomiRank values (Theorem~\ref{thm:automorphism}) --- its positioning is not to discover a new symmetry phenomenon, but to verify that DomiRank satisfies the symmetry axiom of centrality, which is a necessary condition for being a sound centrality index; we then derive DomiRank properties of special structures such as regular graphs and discuss the discriminating power of DomiRank under the Bao--Zhang benchmark (Remark~\ref{rem:discrimination}). This section characterizes the equivalence-class structure of vertices under DomiRank and serves as the foundation for the subsequent entropy and graph-transformation studies.

Beyond the automorphism orbit structure, the question of which vertices receive equal scores has been studied extensively for spectral centralities defined through functions of the adjacency matrix: Ballini and Deniskin \cite{p19} proved that two vertices have equal subgraph centrality for every algebraic parameter if and only if they are cospectral, thereby settling the subgraph-centrality part of Estrada's conjecture on vertex distinction; the closeness and betweenness parts of the same conjecture are known to be false, counterexamples having been given independently by Rombach and Porter \cite{p26} and by Stevanovi\'c, and Deniskin and Benzi surveyed further results and open problems for such centralities \cite{p27}; Boldi, Furia and Prezioso recently characterized the expressivity of the whole family of linear geometric centralities, i.e., which node rankings they can induce \cite{p28}. DomiRank, being a rational function of the adjacency matrix applied to the degree vector (a resolvent-type centrality; see Remark~\ref{rem:interpolation}), satisfies the symmetry axiom by construction. The contribution of this section relative to this literature is twofold: Corollary~\ref{cor:orbit_count} relates the number of distinct DomiRank values to the orbit structure of the automorphism group, and the parameter-dependent ranking evolution established in Section~6 (see also Remark~\ref{rem:discrimination}) shows that the discriminating power of DomiRank is not fixed but evolves continuously with $\sigma$.

\begin{lemma}\label{lemma:regular}
For an $r$-regular graph with $N$ vertices, the normalized DomiRank value of every vertex $v_i \in V(G)$ is $\Gamma_i = \dfrac{1}{N}$.
\end{lemma}

\begin{proof}
Take $\theta=1$ and consider an $r$-regular graph, in which every vertex has degree $r$, i.e., $A \mathbf{1}_{N \times 1} = r \mathbf{1}_{N \times 1}$. Then
	\[
	\begin{aligned}
		\boldsymbol{\Gamma} &= \sigma \left(\sigma A + I\right)^{-1} A \mathbf{1}_{N \times 1} \\
		&= \sigma \left(\sigma A + I\right)^{-1} r \mathbf{1}_{N \times 1} \\
		&= r \sigma \left(\sigma A + I\right)^{-1} \mathbf{1}_{N \times 1}.
	\end{aligned}
	\]

From $A \mathbf{1}_{N \times 1} = r \mathbf{1}_{N \times 1}$, the all-ones vector $\mathbf{1}_{N \times 1}$ is an eigenvector of $A$ with eigenvalue $r$. Therefore
	\[
	\left(\sigma A + I\right) \mathbf{1}_{N \times 1} = \sigma r \mathbf{1}_{N \times 1} + \mathbf{1}_{N \times 1} = \left(\sigma r + 1\right) \mathbf{1}_{N \times 1},
	\]
i.e., $\mathbf{1}_{N \times 1}$ is an eigenvector of $\sigma A + I$ with eigenvalue $\sigma r + 1$, so
\[
\left(\sigma A + I\right)^{-1} \mathbf{1}_{N \times 1} = \frac{1}{\sigma r + 1} \mathbf{1}_{N \times 1}.
\]

Hence
\begin{align}\label{eq:regular_gamma}
\boldsymbol{\Gamma} = \frac{r \sigma}{\sigma r + 1} \mathbf{1}_{N \times 1}.
\end{align}

Thus $\boldsymbol{\Gamma}$ is a constant vector with all entries $\dfrac{r \sigma}{\sigma r + 1}$; after normalization, $\Gamma_i = \dfrac{1}{N}$.
\end{proof}

\begin{theorem}\label{thm:automorphism}
Let $i, j$ be two distinct vertices of a graph $G$. If there exists an automorphism $\psi \in \mathrm{Aut}(G)$ with $\psi(i) = j$, then $\Gamma_i = \Gamma_j$.
\end{theorem}

\begin{proof}
Let $A$ be the adjacency matrix of $G$. Since $\psi$ is an automorphism, there exists a permutation matrix $P$ (satisfying $P^T = P^{-1}$) such that $A = P A P^T$.

From $A = P A P^T$ and $P^T = P^{-1}$ we get $AP = PA$, i.e., $A$ and $P$ commute.

Since $\sigma A + I = P(\sigma A + I)P^T$, taking inverses gives
\[
(\sigma A + I)^{-1} = P(\sigma A + I)^{-1} P^T,
\]
so $(\sigma A + I)^{-1}$ also commutes with $P$:
\[
P(\sigma A + I)^{-1} = (\sigma A + I)^{-1} P.
\]

By equation \eqref{eq:solution}, noting $P\mathbf{1}_{N \times 1} = \mathbf{1}_{N \times 1}$ (a permutation matrix does not change the all-ones vector), we have
\[
\begin{aligned}
	P\boldsymbol{\Gamma} &= \theta \sigma P(\sigma A + I)^{-1} A \mathbf{1}_{N \times 1} \\
	&= \theta \sigma (\sigma A + I)^{-1} P A \mathbf{1}_{N \times 1} \\
	&= \theta \sigma (\sigma A + I)^{-1} A P \mathbf{1}_{N \times 1} \\
	&= \theta \sigma (\sigma A + I)^{-1} A \mathbf{1}_{N \times 1} \\
	&= \boldsymbol{\Gamma}.
\end{aligned}
\]

From $P\boldsymbol{\Gamma} = \boldsymbol{\Gamma}$, $\boldsymbol{\Gamma}$ is invariant under the permutation $P$, i.e., $\Gamma_{\psi(i)} = \Gamma_i$, which means $\Gamma_i = \Gamma_j$.
\end{proof}

\begin{remark}
Theorem~\ref{thm:automorphism} shows that similar vertices (i.e., vertices in the same orbit) have equal DomiRank values. However, the converse is false --- the Frucht graph of Figure~\ref{figure2} is a counterexample. The Frucht graph is a 3-regular graph with 12 vertices and 18 edges; by Lemma~\ref{lemma:regular} its DomiRank vector is $\left[\dfrac{1}{12}, \dfrac{1}{12}, \cdots, \dfrac{1}{12}\right]$, yet the Frucht graph is asymmetric \cite{p29}, its automorphism group containing only the identity.
\end{remark}

\begin{figure}[!h]
	\centering
	\includegraphics{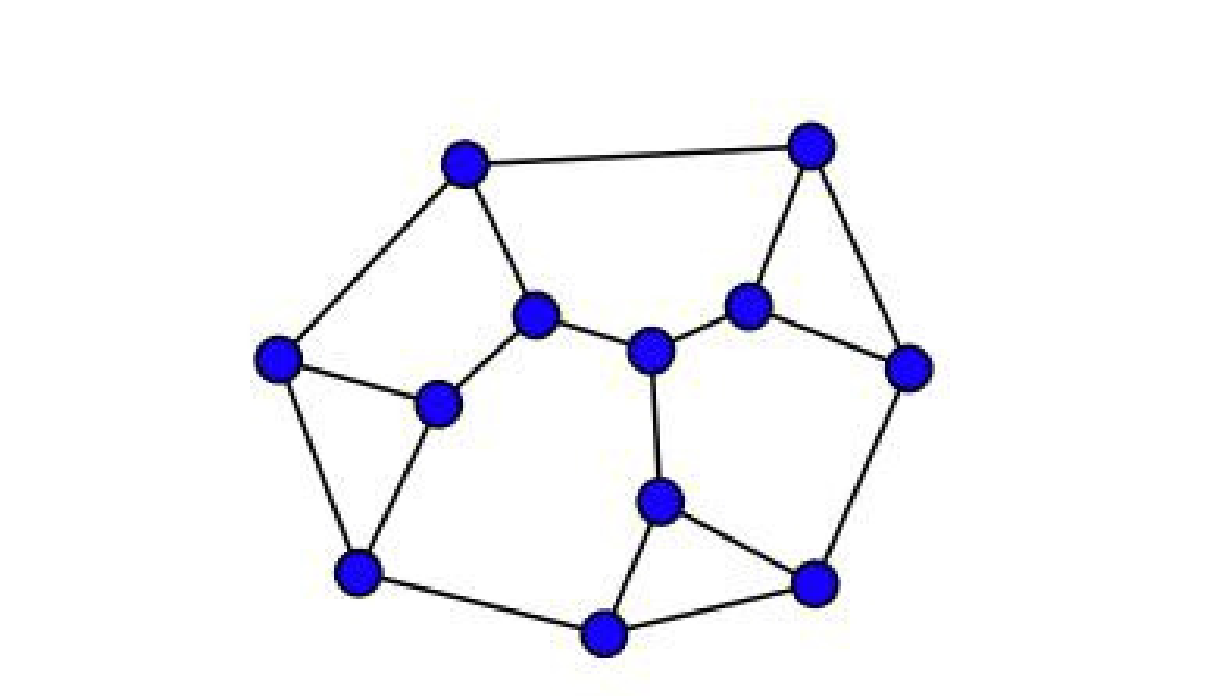}\\
	\caption{The Frucht graph: a 3-regular asymmetric graph with 12 vertices (see the remark after Theorem~\ref{thm:automorphism}).}
	\textbf{Alt text:}\ A bar chart of the DomiRank vector of the 11-vertex star graph at a single competition intensity, with one bar per vertex. The nine leaf vertices receive equal small values and the hub receives a value roughly an order of magnitude larger.
	\label{figure2}
\end{figure}

\begin{corollary}\label{cor:asymmetric}
If all vertices of a graph $G$ have pairwise distinct DomiRank values, then $G$ is asymmetric.
\end{corollary}

\begin{proof}
(Contrapositive.) Suppose a non-identity automorphism $\psi \in \mathrm{Aut}(G)$ exists; then there are vertices $u \neq v$ with $\psi(u) = v$. By Theorem~\ref{thm:automorphism}, $\Gamma_u = \Gamma_v$, contradicting the pairwise distinctness of all DomiRank values.
\end{proof}

\begin{corollary}\label{cor:tree}
If all pendant vertices of a tree $T$ have pairwise distinct DomiRank values, then the automorphism group of $T$ consists only of the identity.
\end{corollary}

\begin{proof}
We first establish a fact: every non-identity automorphism of a tree permutes at least two pendant vertices. Let $\varphi$ be an automorphism of $T$ fixing all pendant vertices. For any vertex $v$ and pendant vertex $u$, automorphisms preserve distances, so $d(\varphi(v), u) = d(\varphi(v), \varphi(u)) = d(v, u)$. In a tree, a vertex is uniquely determined by its distances to all pendant vertices: if $x \neq y$ and $x$ is not pendant, then $x$ has a branch not passing through $y$, along which some pendant vertex $u$ is reached with $d(x, u) < d(y, u)$; if $x$ is pendant, then $d(x, x) = 0 < d(y, x)$. Hence $\varphi(v) = v$ for all vertices, i.e., $\varphi$ is the identity.

Now let $\varphi$ be any non-identity automorphism of $\mathrm{Aut}(T)$. By the fact above, there exist pendant vertices $i$ and $j$ ($i \neq j$) with $\varphi(i) = j$. By Theorem~\ref{thm:automorphism}, $\Gamma_i = \Gamma_j$, contradicting the assumption that all pendant vertices have distinct DomiRank values.
\end{proof}

\begin{definition}\label{def:orbit}
Let $G$ be a graph with automorphism group $\mathrm{Aut}(G)$. The orbit of a vertex $u \in V(G)$, denoted by $u^{\mathrm{Aut}(G)}$ or $[u]$, is
\[
u^{\mathrm{Aut}(G)} = \{\alpha(u): \alpha \in \mathrm{Aut}(G)\}.
\]
\end{definition}

A graph $G$ with a single orbit is called vertex-transitive; that is, for any two vertices $u, v \in V(G)$ there exists an automorphism $\alpha \in \mathrm{Aut}(G)$ with $\alpha(u) = v$.

\begin{theorem}\label{thm:orbit}
Let $L_1, L_2, \dots, L_k$ be all orbits of $\mathrm{Aut}(G)$. For any two vertices $x, y \in L_i$ ($1 \le i \le k$), $\Gamma_x = \Gamma_y$. In particular, if $G$ is vertex-transitive, then $DR_C(G)=1$, where $DR_C(G)$ denotes the number of distinct entries of the DomiRank vector.
\end{theorem}

\begin{proof}
Let $x, y$ belong to the same orbit $L_i$; then there exists an automorphism $\psi$ with $\psi(x) = y$. By Theorem~\ref{thm:automorphism}, $\Gamma_x = \Gamma_y$. Hence all vertices in the same orbit share the same DomiRank value. If $G$ is vertex-transitive, all vertices lie in one orbit, all DomiRank values are equal, and $DR_C(G)=1$.
\end{proof}

\begin{corollary}\label{cor:orbit_count}
Let $m$ be the number of distinct orbits of $G$. Then the DomiRank vector of $G$ has at most $m$ distinct entries.
\end{corollary}

\begin{proof}
By Theorem~\ref{thm:orbit}, vertices in the same orbit have equal DomiRank values; hence $m$ orbits produce at most $m$ distinct DomiRank values.
\end{proof}

\begin{remark}[Discriminating power under the Bao--Zhang benchmark]\label{rem:discrimination}
Within the benchmark framework of \cite{p10}, ``equal values on the same orbit'' is a symmetry axiom that any sound centrality index must satisfy; Theorem~\ref{thm:automorphism} shows that DomiRank passes this soundness check, while the extent to which ``different orbits take different values'' measures the discriminating power of the index. Corollary~\ref{cor:orbit_count} shows that the number of distinct DomiRank entries is bounded above by the number of orbits, and the example of the Frucht graph (the remark after Theorem~\ref{thm:automorphism}; 12 orbits but $DR_C(G) = 1$) shows that a large gap can exist between this upper bound and the actual value: on regular (more generally, harmonic; see Proposition~\ref{prop:invariance}) structures, the normalized DomiRank degenerates to the degree distribution and its discriminating ability is limited. However, on non-harmonic structures, the competition parameter $\sigma$ drives a continuous evolution of the node ranking and produces ranking reversals (see Figure~\ref{figure11} and Table~\ref{tab:top5} of Section~6), endowing DomiRank with a discriminating dimension beyond pure combinatorial symmetry (the orbit structure) --- this is precisely the incremental contribution of DomiRank to discriminating power relative to parameter-free indices such as degree or eigenvector centrality.
\end{remark}

\begin{example}\label{ex:orbit}
In the graph of Figure~\ref{figure1}, vertices $v_1, v_2$ are mapped to each other by an automorphism and belong to the same orbit; $v_4, v_5$ also form an orbit; $v_3$ alone forms an orbit. Hence the number of orbits is 3, and the DomiRank vector has 3 distinct entries, consistent with Corollary~\ref{cor:orbit_count}.

As shown in Figure~\ref{figure3}, the graph is $K_5$ with the two edges $v_1v_4$ and $v_2v_3$ deleted. Its automorphism group has order $8$ and has exactly two orbits, $\{v_1,v_2,v_3,v_4\}$ and $\{v_5\}$: the transpositions $(v_1\,v_2)(v_3\,v_4)$ and $(v_2\,v_3)$ are automorphisms, so $v_1$ and $v_2$ lie in the same orbit, as do $v_2$ and $v_3$, and $v_4$ is reachable from $v_1$ as well. At $\sigma = 0.3$ the DomiRank vector is $\left[\dfrac{27}{62},\; \dfrac{27}{62},\; \dfrac{27}{62},\; \dfrac{27}{62},\; \dfrac{21}{31}\right]$, which contains two distinct entries, one per orbit; the bound of Corollary~\ref{cor:orbit_count} is attained here, and each orbit receives its own value.

The converse situation, in which the bound is not attained because vertices in \emph{different} orbits share a value, requires slightly larger graphs. Figure~\ref{figure15} shows a seven-vertex graph with edge set
{\small
\[
\{v_1v_4,\, v_1v_5,\, v_2v_4,\, v_2v_5,\, v_3v_4,\, v_3v_5,\, v_3v_6,\, v_3v_7,\, v_6v_7\},
\]
}
whose automorphism group has order $8$ and four orbits $\{v_1,v_2\}$, $\{v_3\}$, $\{v_4,v_5\}$, $\{v_6,v_7\}$. At $\sigma = 0.3$ its DomiRank vector is
\[
 \bigl[\tfrac{21}{71},\tfrac{21}{71},\tfrac{51}{71},\tfrac{36}{71},\tfrac{36}{71},\tfrac{21}{71},\tfrac{21}{71}\bigr]:
\]
there are three distinct entries for four orbits, and the two orbits $\{v_1,v_2\}$ and $\{v_6,v_7\}$ both receive the value $21/71$. An exhaustive check over all connected graphs with $N \le 6$ shows that no such instance exists in that range, so seven vertices is the smallest order at which the orbit-count bound of Corollary~\ref{cor:orbit_count} fails to be attained.
\end{example}

\begin{figure}[!h]
	\centering
	\includegraphics[width=0.92\textwidth]{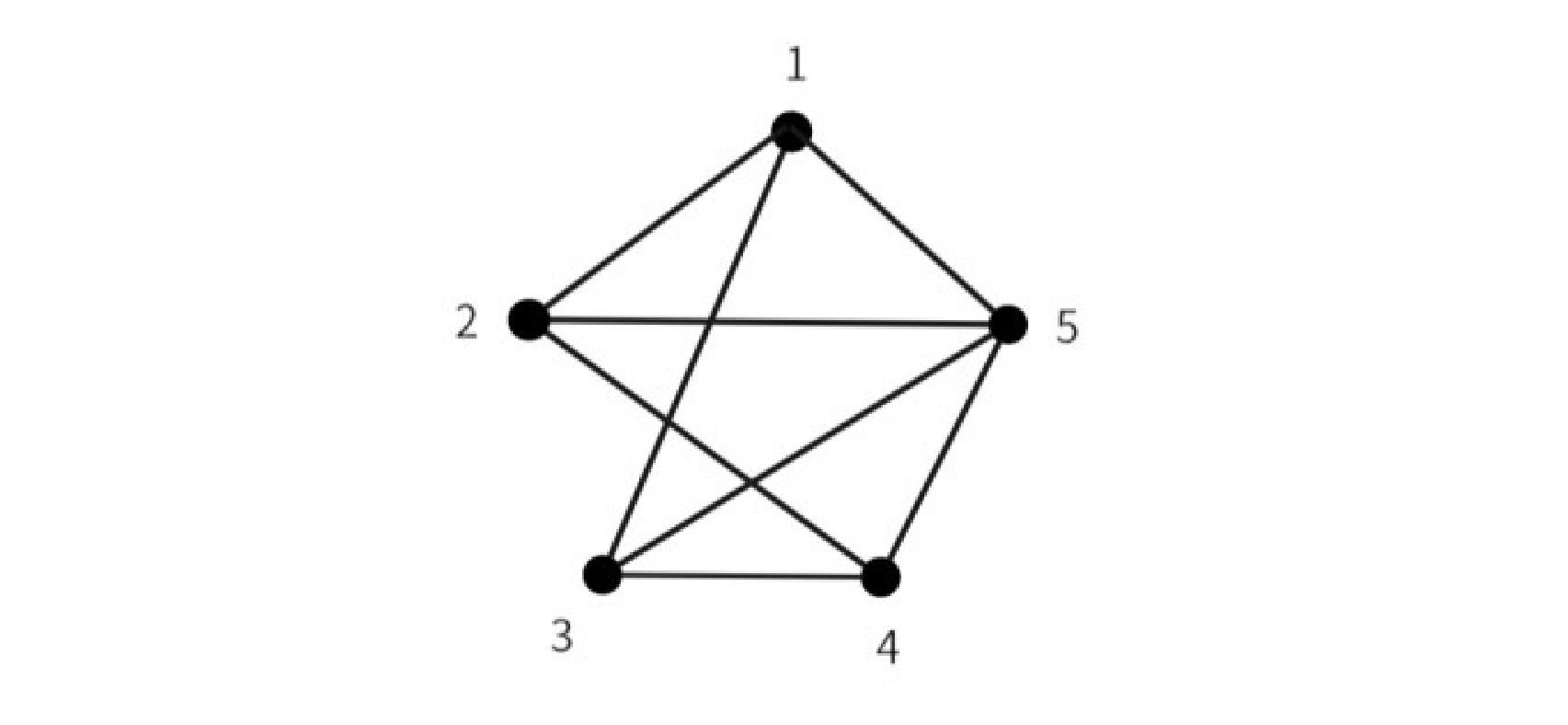}\\
	\caption{The graph $K_5-\{v_1v_4,\,v_2v_3\}$, which has two automorphism orbits $\{v_1,v_2,v_3,v_4\}$ and $\{v_5\}$ and attains the orbit-count bound of Corollary~\ref{cor:orbit_count} with two distinct DomiRank values (Example~\ref{ex:orbit}).}
	\textbf{Alt text:}\ A drawing of the complete graph K5 on five vertices with the two edges v1v4 and v2v3 removed. The graph has two automorphism orbits, {v1,v2,v3,v4} and {v5}, and attains the orbit-count bound with two distinct DomiRank values.
	\label{figure3}
\end{figure}

\begin{figure}[!h]
	\centering
	\includegraphics{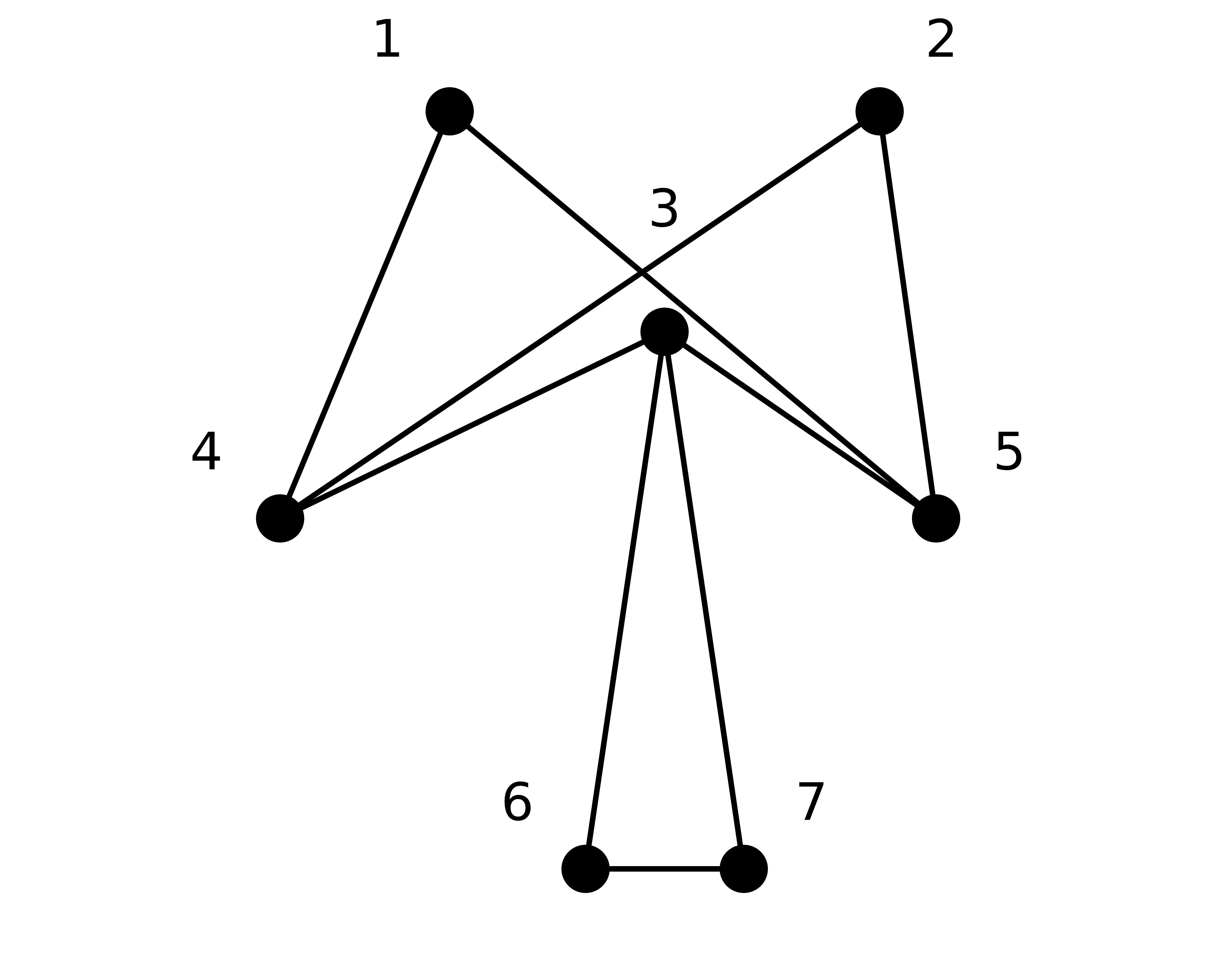}\\
	\caption{A seven-vertex graph with four automorphism orbits $\{v_1,v_2\}$, $\{v_3\}$, $\{v_4,v_5\}$, $\{v_6,v_7\}$ and only three distinct DomiRank values: the orbits $\{v_1,v_2\}$ and $\{v_6,v_7\}$ share the value $21/71$ at $\sigma=0.3$, so the bound of Corollary~\ref{cor:orbit_count} is not attained (Example~\ref{ex:orbit}).}
	\textbf{Alt text:}\ A drawing of a seven-vertex graph with edges v1v4, v1v5, v2v4, v2v5, v3v4, v3v5, v3v6, v3v7 and v6v7. It has four automorphism orbits but only three distinct DomiRank values, so the orbit-count bound is not attained; this is the smallest such example.
	\label{figure15}
\end{figure}

\begin{lemma}\label{lemma:regular_iff}
Let $G$ be a connected graph and take $\theta = 1$. Then $G$ is regular if and only if $\boldsymbol{\Gamma} = \lambda \vec{j}$ for some $\lambda \in \mathbb{R}$, where $\vec{j} = [1,1,\dots,1]^T$.
\end{lemma}
\begin{proof}
\noindent\textbf{Sufficiency:} If $G$ is $r$-regular, Lemma~\ref{lemma:regular} gives $\boldsymbol{\Gamma} = \dfrac{r \sigma}{\sigma r + 1} \vec{j}$, i.e., $\boldsymbol{\Gamma} = \lambda \vec{j}$ with $\lambda = \dfrac{r \sigma}{\sigma r + 1}$.

\noindent\textbf{Necessity:} If $\boldsymbol{\Gamma} = \lambda \vec{j}$, all vertices have equal DomiRank values, i.e., $\Gamma_i = \lambda$ for all $i$.

From equation \eqref{eq:node_solution} (with $\theta = 1$):
\[
\Gamma_i = \sigma \sum_{k \in N(i)} \left( 1 - \Gamma_k \right).
\]

Substituting $\Gamma_i = \lambda$ and $\Gamma_k = \lambda$,
\[
\lambda = \sigma (1 - \lambda) k_i,
\]
where $k_i$ is the degree of vertex $i$. Solving,
\begin{align}\label{eq:degree}
k_i = \dfrac{\lambda}{\sigma (1 - \lambda)}.
\end{align}

Since $\lambda$ and $\sigma$ are constants, all degrees $k_i$ are equal, so $G$ is regular.
\end{proof}

\begin{theorem}\label{thm:two_orbits}
Let $G$ be a connected graph with exactly two orbits. Then either $G$ is regular or $DR_C(G) = 2$.
\end{theorem}

\begin{proof}
Since $G$ has two orbits, Corollary~\ref{cor:orbit_count} gives $DR_C(G) \le 2$. If $DR_C(G) = 1$, all vertices have equal DomiRank values and Lemma~\ref{lemma:regular_iff} implies that $G$ is regular; if $DR_C(G) = 2$, the conclusion holds.
\end{proof}

\begin{lemma}\label{lemma:same_neighbor}
Let $G$ be a connected graph and $i, j$ two distinct vertices with $N(i) = N(j)$. Then $\Gamma_i = \Gamma_j$.
\end{lemma}

\begin{proof}
Since $N(i) = N(j)$, equation \eqref{eq:node_solution} yields
\begin{align*}
\Gamma_i &= \sigma \sum_{k \in N(i)} \left( \theta - \Gamma_k \right)
 = \sigma \sum_{k \in N(j)} \left( \theta - \Gamma_k \right) \\
&= \Gamma_j.
\end{align*}
Hence $\Gamma_i = \Gamma_j$.
\end{proof}

\begin{lemma}\label{lemma:subset_neighbor}
Let $i, j$ be adjacent vertices of a graph $G$ with $\sigma < 1$. If $N(i)\setminus\{j\} \subsetneq N(j)\setminus\{i\}$ and $\theta - \Gamma_k > 0$ for all $k \in \big(N(j)\setminus\{i\}\big)\setminus\big(N(i)\setminus\{j\}\big)$, then $\Gamma_i < \Gamma_j$.
\end{lemma}

\begin{proof}
Since $i, j$ are adjacent, $j \in N(i)$ and $i \in N(j)$. By equation \eqref{eq:node_solution},
\begin{align*}
\Gamma_i &= \sigma\left(\theta - \Gamma_j\right)
 + \sigma\sum_{k \in N(i)\setminus\{j\}} \left(\theta - \Gamma_k\right), \notag\\
\Gamma_j &= \sigma\left(\theta - \Gamma_i\right)
 + \sigma\sum_{k \in N(j)\setminus\{i\}} \left(\theta - \Gamma_k\right).
\end{align*}

Subtracting the two equations and writing $D = \big(N(j)\setminus\{i\}\big)\setminus\big(N(i)\setminus\{j\}\big)$, from $N(i)\setminus\{j\} \subsetneq N(j)\setminus\{i\}$ we have $D \neq \emptyset$, and
\[
\begin{split}
\Gamma_j - \Gamma_i
&= \sigma\left(\Gamma_j - \Gamma_i\right)
 + \sigma\left[\sum_{k \in N(j)\setminus\{i\}}\left(\theta - \Gamma_k\right)
 - \sum_{k \in N(i)\setminus\{j\}}\left(\theta - \Gamma_k\right)\right] \\
&= \sigma\left(\Gamma_j - \Gamma_i\right) + \sigma\sum_{k \in D}\left(\theta - \Gamma_k\right).
\end{split}
\]

Rearranging,
\begin{align}\label{eq:subset}
(1-\sigma)\left(\Gamma_j - \Gamma_i\right) = \sigma\sum_{k \in D}\left(\theta - \Gamma_k\right).
\end{align}

By assumption, every $k \in D$ satisfies $\theta - \Gamma_k > 0$, so the right-hand side is strictly positive; since $\sigma < 1$, $\Gamma_j - \Gamma_i > 0$, i.e., $\Gamma_i < \Gamma_j$.
\end{proof}

\begin{remark}
In particular, if adjacent vertices $i, j$ satisfy $N(i)\setminus\{j\} = N(j)\setminus\{i\}$, then the difference set $D$ in \eqref{eq:subset} is empty, and hence $\Gamma_i = \Gamma_j$. This can be viewed as the counterpart of Lemma~\ref{lemma:same_neighbor} for adjacent vertices.
\end{remark}

\begin{lemma}\label{lemma:pendant}
Let $i$ be a pendant vertex of $G$ (degree 1), adjacent to vertex $j$ with $d_j \ge 2$. If $\sigma < 1$ and $\theta - \Gamma_k > 0$ for all $k \in N(j)\setminus\{i\}$, then $\Gamma_i < \Gamma_j$.
\end{lemma}

\begin{proof}
Since $i$ is pendant and adjacent to $j$, equation \eqref{eq:node_solution} gives
\[
\Gamma_i = \sigma \left( \theta - \Gamma_j \right).
\]

For vertex $j$ ($d_j \ge 2$), let $S = \sum_{k \in N(j)\setminus\{i\}} (\theta - \Gamma_k)$; then
\[
\Gamma_j = \sigma \left( \theta - \Gamma_i \right) + \sigma S.
\]

Computing $\Gamma_j - \Gamma_i$:
\[
\Gamma_j - \Gamma_i = \sigma(\theta - \Gamma_i) + \sigma S - \sigma(\theta - \Gamma_j) = \sigma(\Gamma_j - \Gamma_i) + \sigma S.
\]

Rearranging,
\[
(1 - \sigma)(\Gamma_j - \Gamma_i) = \sigma S.
\]

By assumption, $\theta - \Gamma_k > 0$ for $k \in N(j)\setminus\{i\}$, so $S > 0$; and since $\sigma < 1$,
\[
\Gamma_j - \Gamma_i = \frac{\sigma S}{1 - \sigma} > 0,
\]
i.e., $\Gamma_i < \Gamma_j$.
\end{proof}

\begin{remark}\label{rem:sigma_scale}
The condition $\sigma < 1$ in Lemmas~\ref{lemma:subset_neighbor} and \ref{lemma:pendant} is stated in absolute units, while the natural scale of $\sigma$ is spectral, $(0, -1/\lambda_N)$. On graphs with $|\lambda_N| > 1$ the condition $\sigma < 1$ covers the entire admissible range; on graphs with $|\lambda_N| < 1$ it restricts attention to an initial portion of the range, and the conclusion may continue to hold beyond $\sigma = 1$ whenever $1 - \sigma$ is replaced by the appropriate sign analysis of \eqref{eq:subset}.
\end{remark}

\begin{remark}\label{rem:pendant_sharp}
The threshold condition in Lemma~\ref{lemma:pendant} --- ``$\Gamma_k < \theta$ for all $k \in N(j)\setminus\{i\}$'' --- cannot be removed: DomiRank values may exceed the dominance threshold $\theta$ (when $\sigma$ approaches the convergence upper bound $-1/\lambda_N$, the values of high-degree vertices diverge), in which case the conclusion may fail. For example, take a 13-vertex tree: the central vertex $h$ is connected to 10 pendant vertices and to vertex $j$, which in turn is connected to the pendant vertex $i$. For $\sigma = 0.27$ (inside the convergence region), numerical computation gives $\Gamma_h \approx 10.942 > \theta$, $\Gamma_j \approx -2.683$, and $\Gamma_i = \sigma(\theta - \Gamma_j) \approx 0.994$; here $\Gamma_i > \Gamma_j$, i.e., the DomiRank value of the pendant vertex is larger than that of its unique neighbor.
\end{remark}

\section{DomiRank Entropy}

Building on the characterization of vertex equivalence by automorphisms, this section turns to the uniformity measure of the importance distribution and studies the entropy of the DomiRank distribution and its variation with the competition parameter. The Shannon entropy \cite{p30} is generally defined as
\begin{align}\label{eq:entropy}
H(\boldsymbol{\Gamma}) = -\sum_{i=1}^N \Gamma_i \log_2 \Gamma_i,
\end{align}
where $\Gamma_i$ are the normalized DomiRank values satisfying $\sum_{i=1}^N \Gamma_i = 1$; the definition applies on the parameter range where all DomiRank entries are positive, so that the normalized values form a probability distribution (parameters at which non-positive entries arise are treated through the effective, truncated distribution of Section~\ref{sec:sigma}).

\paragraph{Organization and strength of the results.}
The results of this section are deliberately presented in three layers of decreasing rigor, and we state the layering up front so that the scope of each statement is not overread.
\begin{itemize}
\item[(L1)] \emph{Local criterion, fully rigorous.} The initial entropy derivative is computed in closed form and its sign characterized exactly (Proposition~\ref{prop:initial_entropy}); this settles the behaviour as $\sigma \to 0$ on all connected graphs, not only asymptotically.
\item[(L2)] \emph{Sufficient-condition theorem, rigorous but with limited coverage.} Theorem~\ref{thm:entropy_decrease} gives conditions (i)--(iii) guaranteeing $dH/d\sigma < 0$ at a parameter, proved for $\sigma < 1/\lambda_1$. We state plainly that these conditions cover less than $10\%$ of the sampled parameter range of the four real networks of Section~\ref{sec:sigma}; the theorem is a guarantee on an explicitly delimited region, not a description of typical behaviour.
\item[(L3)] \emph{Conjecture.} That $H$ decreases with $\sigma$ without extra conditions is Conjecture~\ref{conj:entropy}, supported by exhaustive enumeration on all connected graphs with $N \le 7$ and by numerics on real networks, and refuted in its unconditional form by an assortative counterexample.
\end{itemize}
Readers who only want the monotone behaviour may take the empirical evidence in Section~\ref{sec:sigma} as evidence for Conjecture~\ref{conj:entropy}, and Theorem~\ref{thm:entropy_decrease} as one sufficient route to it; the general case remains open.

\begin{theorem}\label{thm:entropy_max}
For a connected graph $G$ with $N$ vertices, at any admissible parameter $\sigma$ at which the DomiRank vector is entrywise positive, the DomiRank entropy attains its maximum $\log_2 N$ if and only if $G$ is regular.
\end{theorem}

\begin{proof}
\noindent \textbf{Sufficiency:}
If $G$ is regular, Lemma~\ref{lemma:regular} gives equal DomiRank values for all vertices, and after normalization $\Gamma_1 = \Gamma_2 = \dots = \Gamma_N = \dfrac{1}{N}$. Hence
\begin{align}\label{eq:max_entropy}
H(\boldsymbol{\Gamma}) = -\sum_{i=1}^N \frac{1}{N} \log_2 \frac{1}{N} = \log_2 N.
\end{align}

\noindent \textbf{Necessity:}
The entropy $H(\boldsymbol{\Gamma}) = -\sum_{i=1}^N \Gamma_i \log_2 \Gamma_i$ is strictly concave under the constraint $\sum \Gamma_i = 1$, so its maximum is attained at the unique critical point. Construct the Lagrangian
\[
L = -\sum_{i=1}^N \Gamma_i \log_2 \Gamma_i + \lambda \left( 1 - \sum_{i=1}^N \Gamma_i \right).
\]

Differentiating with respect to $\Gamma_j$ and setting to zero:
\[
\frac{\partial L}{\partial \Gamma_j} = -\left( \log_2 \Gamma_j + \frac{1}{\ln 2} \right) - \lambda = 0.
\]

All $\Gamma_j$ satisfy the same equation; by strict concavity the maximizer satisfies
\[
\Gamma_1 = \Gamma_2 = \dots = \Gamma_N.
\]

Together with the constraint $\sum \Gamma_i = 1$, we get $\Gamma_1 = \Gamma_2 = \dots = \Gamma_N = \dfrac{1}{N}$. By Lemma~\ref{lemma:regular_iff}, the connected graph $G$ is regular.
\end{proof}

\subsection{Influence of the parameter $\sigma$ on the entropy}

\begin{lemma}\label{lemma:neumann}
Let $A$ be a square matrix with $\rho(A) < 1$ ($\rho$ the spectral radius). Then $(I+A)^{-1}$ exists and is given by the convergent Neumann series
\begin{align}\label{eq:neumann}
(I+A)^{-1} = \sum_{m=0}^{\infty} (-1)^m A^m
= I - A + A^2 - A^3 + \cdots
\end{align}
\end{lemma}

\begin{theorem}\label{thm:sigma_zero}
As $\sigma \to 0$, DomiRank converges to degree centrality; hence $H(\boldsymbol{\Gamma}) \to H(\text{degree})$.
\end{theorem}

\begin{proof}
By Lemma~\ref{lemma:neumann}, for sufficiently small $\sigma$,
\begin{align*}
(\sigma A + I)^{-1}
= I - \sigma A + \sigma^2 A^2 - \sigma^3 A^3 + \cdots
\end{align*}

Taking $\theta = 1$,
\[
\boldsymbol{\Gamma} = \sigma \left(I - \sigma A + \sigma^2 A^2 - \cdots \right) A \mathbf{1}_N = \sigma A \mathbf{1}_N - \sigma^2 A^2 \mathbf{1}_N + O(\sigma^3).
\]

As $\sigma \to 0$, retaining the leading term and neglecting higher-order terms,
\[
\boldsymbol{\Gamma} \approx \sigma A \mathbf{1}_N = \sigma \boldsymbol{k},
\]
where $\boldsymbol{k} = A \mathbf{1}_N$ is the degree vector. After normalization, the DomiRank distribution approaches the degree distribution, so $H(\boldsymbol{\Gamma}) \approx H(\boldsymbol{k})$.
\end{proof}

\begin{proposition}[Initial entropy derivative]\label{prop:initial_entropy}
Let $G$ be a connected graph with degrees $k_1, \dots, k_N$, $2m = \sum_i k_i$, and normalized DomiRank entropy $H(\sigma)$. Then
\begin{align}\label{eq:initial_derivative}
\lim_{\sigma \to 0^+} \frac{dH}{d\sigma}
&= \frac{1}{2m} \sum_{\{i,j\} \in E} \left( k_i \log_2 \frac{k_j}{2m} + k_j \log_2 \frac{k_i}{2m} \right) \notag\\
&\quad - \frac{\sum_i k_i^2}{2m \sum_i k_i} \sum_i k_i \log_2 \frac{k_i}{2m},
\end{align}
which equals the edge covariance $\mathrm{Cov}_{\{i,j\} \in E}\!\big( k_i, \log_2 k_j \big)$ over a uniformly random (oriented) edge. In particular:
\begin{enumerate}
\item[(i)] the limit is $0$ for regular graphs --- more generally for harmonic graphs, where $(A\boldsymbol{k})_i = \rho\, k_i$ with $\rho = \sum_i k_i^2 / \sum_i k_i$, consistent with the $\sigma$-invariance of Proposition~\ref{prop:invariance};
\item[(ii)] the entropy initially decreases if and only if this edge covariance is negative, i.e., if the graph is disassortative in the degree--log-degree sense.
\end{enumerate}
\end{proposition}

\begin{proof}
Write $a_i = (A\boldsymbol{k})_i = \sum_{j \in N(i)} k_j$ and $S(\sigma) = \sum_i \Gamma_i(\sigma)$. By the expansion of Theorem~\ref{thm:sigma_zero}, $\Gamma_i(\sigma) = \sigma k_i - \sigma^2 a_i + O(\sigma^3)$ and $S(\sigma) = \sigma (2m) - \sigma^2 \sum_i k_i^2 + O(\sigma^3)$ (using $\sum_i a_i = \sum_i k_i^2$). Hence, with $P_i(\sigma) = \Gamma_i(\sigma)/S(\sigma)$,
\[
P_i(\sigma) = \frac{k_i}{2m} + \sigma \left( \frac{k_i \sum_j k_j^2}{(2m)^2} - \frac{a_i}{2m} \right) + O(\sigma^2).
\]
Differentiating $H(\sigma) = -\sum_i P_i(\sigma) \log_2 P_i(\sigma)$ and using $\sum_i P_i'(0) = 0$ gives
\begin{gather*}
\lim_{\sigma \to 0^+} \frac{dH}{d\sigma} \\
= -\sum_i P_i'(0) \log_2 \frac{k_i}{2m} \\
= \frac{1}{2m} \left( \sum_i a_i \log_2 \frac{k_i}{2m}
- \frac{\sum_j k_j^2}{2m} \sum_i k_i \log_2 \frac{k_i}{2m} \right),
\end{gather*}
and, since
\begin{align*}
\sum_i a_i \log_2 \frac{k_i}{2m}
&= \sum_i \sum_{j \in N(i)} k_j \log_2 \frac{k_i}{2m} \\
&= \sum_{\{i,j\} \in E} \left( k_j \log_2 \frac{k_i}{2m}
   + k_i \log_2 \frac{k_j}{2m} \right),
\end{align*}
this proves \eqref{eq:initial_derivative}. For the covariance identity, a uniformly random oriented edge has endpoint law $\Pr(i) = k_i/(2m)$, under which
\[
 \begin{aligned}
 \mathbb{E}[k_i] &= \sum_i k_i^2/(2m), \\
 \mathbb{E}[\log_2 k_j] &= \sum_j k_j \log_2 k_j/(2m), \\
 \mathbb{E}[k_i \log_2 k_j] &= \sum_{\{i,j\} \in E}
   \frac{k_i \log_2 k_j + k_j \log_2 k_i}{2m};
 \end{aligned}
\]
substituting shows that \eqref{eq:initial_derivative} equals $\mathrm{Cov}(k_i, \log_2 k_j)$ (the additive constant $-\log_2(2m)$ inside the logarithm does not affect the covariance). (i) On a harmonic graph $a_i = \rho k_i$ with
\[
 \rho = \sum_j k_j^2/(2m) \cdot (2m) \Big/ \sum_j k_j = \sum_j k_j^2 \Big/ \sum_j k_j,
\]
and substitution into \eqref{eq:initial_derivative} gives $0$; the regular case is the special case $\rho = k$. (ii) is immediate from the covariance identity.
\end{proof}

\begin{remark}[An assortative counterexample to unconditional monotonicity]\label{rem:assortative_counter}
The sign in Proposition~\ref{prop:initial_entropy} is genuinely graph-dependent. Let $G$ be the 7-vertex graph obtained from a $K_5$ and a $K_2$ joined by a single bridge (degree sequence $5, 4, 4, 4, 4, 2, 1$). Its edge covariance equals $+0.2149 > 0$, and indeed the entropy \emph{increases} on the initial interval: $H(0.001) = 2.6848 < H(0.01) = 2.6867$, rising to $2.7073$ at $\sigma = 0.3$ (all DomiRank entries remain positive at least up to $\sigma = 0.35$, and $\sigma^* \approx 0.5805$, so the theoretical entropy is well defined throughout). Hence the unconditional monotonicity statement is false, and Conjecture~\ref{conj:entropy} below is formulated for structurally disassortative graphs. All four real-world networks of Section~6 have negative edge covariance (Karate $-2.8403$, Dolphins $-0.0961$, LesMis $-1.5553$, Polbooks $-0.6807$); on these networks the initial decrease observed in Figure~\ref{figure12} is therefore rigorous, not only numerical.
\end{remark}

\begin{lemma}\label{lemma:polarization}
Let $G$ be a connected graph with normalized DomiRank distribution $P(\sigma)$.
\begin{enumerate}
\item[(i)] If $G$ is harmonic, then $P(\sigma)$ is independent of $\sigma$ on $\left(0, -1/\lambda_N\right)$ and the DomiRank entropy $H$ is constant on this interval (Proposition~\ref{prop:invariance}).
\item[(ii)] If $G$ is non-harmonic, there is a finite set $Z \subset \left(0, -1/\lambda_N\right)$ such that for every $\sigma \in \left(0, -1/\lambda_N\right) \setminus Z$ there exist vertices $k$ and $l$ with $\dfrac{dP_k}{d\sigma} > 0$ and $\dfrac{dP_l}{d\sigma} < 0$; and on each subinterval of $\left(0, -1/\lambda_N\right) \setminus Z$ the function $P$ is strictly monotone in at least one coordinate in each direction --- that is, at least one coordinate is strictly increasing and at least one is strictly decreasing --- so that for any $\sigma_1 < \sigma_2$ lying in the same such subinterval there is a vertex $k$ with $P_k(\sigma_2) > P_k(\sigma_1)$ and a vertex $l$ with $P_l(\sigma_2) < P_l(\sigma_1)$.
\end{enumerate}
\end{lemma}

\begin{proof}
(i) follows directly from Proposition~\ref{prop:invariance}: on harmonic graphs the normalized DomiRank distribution is $\sigma$-invariant, hence the entropy is constant.

(ii) Since $G$ is non-harmonic, Proposition~\ref{prop:invariance} shows that $P(\sigma)$ is not constant in $\sigma$. Each entry of $\boldsymbol{\Gamma}(\sigma) = \theta\sigma(\sigma A + I)^{-1}A\mathbf{1}$ is a rational function of $\sigma$ (the resolvent $(\sigma A + I)^{-1}$ has rational entries), so $P(\sigma)$ is a rational vector function; hence $P'(\sigma)$ is a rational vector function that is not identically zero, and it has finitely many zeros in the admissible interval. At any $\sigma$ which is not such a zero, $P'(\sigma) \neq 0$; since $\sum_{i=1}^N P_i'(\sigma) = 0$, the entries of $P'(\sigma)$ cannot all have the same sign, so there exist $k, l$ with $P_k'(\sigma) > 0$ and $P_l'(\sigma) < 0$. Between two consecutive zeros the sign of each entry of $P'(\sigma)$ is constant; since $P' \neq 0$ there, integrating over such a subinterval gives $P_k(\sigma_2) > P_k(\sigma_1)$ and $P_l(\sigma_2) < P_l(\sigma_1)$ for any $\sigma_1 < \sigma_2$ inside it, which is exactly the assertion of (ii). We note that the monotonicity established here is \emph{piecewise} in $\sigma$; whether $P$ is injective on the whole admissible interval is not claimed, since a piecewise monotone function may in principle take the same value in two different subintervals.
\end{proof}

\begin{theorem}\label{thm:entropy_decrease}
Let $G$ be a connected non-harmonic graph (the harmonic case is trivial: by Lemma~\ref{lemma:polarization}(i) the entropy is constant). At a fixed admissible parameter $\sigma$, consider the vertex sets $S^+$ and $S^-$ of vertices whose DomiRank values increase and decrease, respectively, with $\sigma$, and let $t \in S^+$ satisfy $P_t = \min\{P_k : k \in S^+\}$. If for all $i \in S^-$ one of the following conditions holds, then $\dfrac{dH}{d\sigma} < 0$ at this parameter; consequently $H(\sigma_1) \ge H(\sigma_2)$ for parameters $\sigma_1 < \sigma_2$ within the range where the conditions apply (conditions (i)--(iii) are rigorous for $\sigma < 1/\lambda_1$, the range where the second-order expansion underlying them is valid, and are a sufficient-condition heuristic for larger $\sigma$; see the discussion following the proof).

	\begin{enumerate}
		\item[(i)] If $k_t - k_i < 0$ and $(A\boldsymbol{k})_t - (A\boldsymbol{k})_i < 0$, then for
		\[
		 \sigma \in \left( \frac{k_t - k_i}{(A\boldsymbol{k})_t - (A\boldsymbol{k})_i},
		 \frac{1}{-\lambda_N} \right), \quad \Gamma_t > \Gamma_i.
		\]

		\item[(ii)] If $k_t - k_i > 0$ and $(A\boldsymbol{k})_t - (A\boldsymbol{k})_i > 0$, then for
		\[
		 \sigma \in \left( 0, \frac{k_t - k_i}{(A\boldsymbol{k})_t - (A\boldsymbol{k})_i} \right),
		 \quad \Gamma_t > \Gamma_i.
		\]

		\item[(iii)] If $k_t - k_i > 0$ and $(A\boldsymbol{k})_t - (A\boldsymbol{k})_i < 0$, then
		$\Gamma_t > \Gamma_i$ holds for all $\sigma \in \left( 0, \frac{1}{-\lambda_N} \right)$.
	\end{enumerate}
\end{theorem}

\begin{proof}
Differentiating $H = -\sum_{i=1}^N P_i \log_2 P_i$ with respect to $\sigma$:
\[
\frac{dH}{d\sigma} = -\sum_{i=1}^N \left( \frac{dP_i}{d\sigma} \log_2 P_i + \frac{dP_i}{d\sigma} \cdot \frac{1}{\ln 2} \right).
\]

Since $\sum_{i=1}^N P_i = 1$, we have $\sum_{i=1}^N \frac{dP_i}{d\sigma} = 0$, and substituting into the above gives
\begin{align}\label{eq:dHdsigma}
\frac{dH}{d\sigma} = -\sum_{i=1}^N \frac{dP_i}{d\sigma} \log_2 P_i.
\end{align}

By Lemma~\ref{lemma:polarization}(ii), at the parameter $\sigma$ under consideration (chosen away from the finitely many exceptional points) there exists a set $S^+ = \{k_1, \dots, k_m\}$ ($m \ge 1$) with $\dfrac{dP_{k_i}}{d\sigma} > 0$, and the remaining set $S^-$ satisfies $\dfrac{dP_i}{d\sigma} \le 0$ (vertices with zero derivative contribute nothing to \eqref{eq:dHdsigma}). Splitting \eqref{eq:dHdsigma},
\begin{align}\label{eq:split}
\frac{dH}{d\sigma} = -\sum_{l \in S^+} \frac{dP_l}{d\sigma} \log_2 P_l - \sum_{i \in S^-} \frac{dP_i}{d\sigma} \log_2 P_i.
\end{align}

Let $P_t = \min\{P_{k_1}, P_{k_2}, \dots, P_{k_m}\}$; then $\log_2 P_t \le \log_2 P_l$ ($l \in S^+$), and since $\dfrac{dP_l}{d\sigma} > 0$,
\begin{align}\label{eq:bound}
\frac{dH}{d\sigma} \le -\sum_{l \in S^+} \frac{dP_l}{d\sigma} \log_2 P_t - \sum_{i \in S^-} \frac{dP_i}{d\sigma} \log_2 P_i.
\end{align}

Since $\displaystyle\sum_{l \in S^+} \frac{dP_l}{d\sigma} = -\sum_{i \in S^-} \frac{dP_i}{d\sigma}$, formula \eqref{eq:bound} simplifies to
\[
\frac{dH}{d\sigma} \le \sum_{i \in S^-} \frac{dP_i}{d\sigma} \log_2 P_t - \sum_{i \in S^-} \frac{dP_i}{d\sigma} \log_2 P_i = \sum_{i \in S^-} \frac{dP_i}{d\sigma} \log_2 \frac{P_t}{P_i}.
\]

Comparing $P_t$ with $P_i$ is equivalent to comparing $\Gamma_t$ with $\Gamma_i$ (normalization preserves the ordering whenever the normalizing sum is positive). By the second-order expansion of Theorem~\ref{thm:sigma_zero}, which is rigorous for $\sigma$ below the Neumann-series convergence radius $1/\lambda_1$ of $(\sigma A + I)^{-1}$, $\boldsymbol{\Gamma} \approx \sigma A \mathbf{1}_N - \sigma^2 A^2 \mathbf{1}_N$, hence
\begin{align}\label{eq:gamma_diff}
\Gamma_t - \Gamma_i \approx \sigma (k_t - k_i) - \sigma^2 \left[ (A\boldsymbol{k})_t - (A\boldsymbol{k})_i \right].
\end{align}

Under conditions (i), (ii), (iii), $\Gamma_t > \Gamma_i$, i.e., $P_t > P_i$, so $\log_2 \dfrac{P_t}{P_i} > 0$. Moreover, for $i \in S^-$, $\dfrac{dP_i}{d\sigma} \le 0$, and Lemma~\ref{lemma:polarization}(ii) ensures that at least one of these derivatives is strictly negative; since $P_t > P_i$ for all $i \in S^-$, every term $\dfrac{dP_i}{d\sigma} \log_2 \dfrac{P_t}{P_i}$ is non-positive and at least one is strictly negative, and therefore
\[
\frac{dH}{d\sigma} < 0,
\]
i.e., the entropy decreases with the competition parameter on the range where the conditions apply.
\end{proof}

We stress that the intervals appearing in conditions (i)--(iii) are derived from the second-order expansion \eqref{eq:gamma_diff}; the expansion is rigorous only for $\sigma < 1/\lambda_1$ (the Neumann-series convergence radius), and for larger $\sigma$ the conditions should be read as a sufficient-condition heuristic. Moreover, even within the provable range the conditions are rarely satisfied: on the four real-world networks of Section~6 fewer than $10\%$ of the relevant vertex pairs meet them in the low-competition region (Remark~\ref{rem:entropy_condition}), so Theorem~\ref{thm:entropy_decrease} by itself covers only a fraction of the regime in which monotonicity is observed; the initial segment of that regime is settled rigorously and independently by Proposition~\ref{prop:initial_entropy}. The numerical evidence of Remark~\ref{rem:entropy_condition} and Section~6 indicates that the monotonicity holds well beyond the provable range; we pose the refined unconditional version as Conjecture~\ref{conj:entropy}.

\begin{remark}[Empirical examination of the conditions]\label{rem:entropy_condition}
Conditions (i)--(iii) of Theorem~\ref{thm:entropy_decrease} are casewise sufficient conditions. We examined them pointwise on the four real-world networks of Section~6 (for each network, 59 adjacent $\sigma$ sampling intervals; for each interval we determined $S^+$, $S^-$ and $t$, and checked all $(t, i)$ vertex pairs): the fraction of vertex pairs satisfying conditions (i)--(iii) is below $10\%$ in the low-competition region ($\sigma/\sigma^* < 0.2$) and ranges between $42\%$ and $100\%$ in the high-competition region ($\sigma/\sigma^* > 0.8$); that is, the sufficient conditions are not universally satisfied on real networks. Nevertheless, the normalized entropies of all four networks decrease strictly monotonically at every sampling point (Figure~\ref{figure12}). This contrast indicates that the conditions of Theorem~\ref{thm:entropy_decrease} are conservative; note, however, that the unconditional monotonicity statement is false in general (Remark~\ref{rem:assortative_counter}), so the conjecture must be restricted accordingly --- we formalize the refined version as Conjecture~\ref{conj:entropy}.
\end{remark}

\begin{conjecture}\label{conj:entropy}
Let $G$ be a connected graph whose degree--log-degree edge covariance \eqref{eq:initial_derivative} is non-positive (structurally disassortative graphs; the restriction is necessary by Remark~\ref{rem:assortative_counter}). For all $0 < \sigma_1 < \sigma_2 < -1/\lambda_N$ such that the theoretical entropy is defined throughout $[\sigma_1, \sigma_2]$ (i.e., all entries of $\Gamma(\sigma)$ are positive there), the normalized DomiRank entropy satisfies $H(\sigma_1) \ge H(\sigma_2)$; moreover, on the domain where $H$ is defined, $H$ is constant if and only if $G$ is harmonic. For parameters at which $\Gamma(\sigma)$ has non-positive entries the theoretical entropy is undefined; there the monotonicity statement is understood for the effective (truncated) distribution of Section~6, for which the numerical evidence is reported in Figure~\ref{figure12}.
\end{conjecture}

\begin{algorithm}[h]
	\caption{Core computation: DomiRank vector and entropy on a relative $\sigma$ grid}
	\label{alg:grid}
	\begin{algorithmic}[1]
		\REQUIRE Network edge lists Source, Target
		\ENSURE DomiRank vector $\boldsymbol{\Gamma}(\sigma)$ and normalized entropy $H(\sigma)$ at each grid point

		\STATE Construct graph $G$ with edges from Source to Target
		\STATE $A \leftarrow$ adjacency matrix of $G$
		\STATE $N \leftarrow$ dimension of $A$

		\STATE \textit{Compute the admissible range of $\sigma$:}
		\STATE $\text{lambda} \leftarrow$ eigenvalues of $A$ (real parts)
		\STATE $\text{valid\_lambda} \leftarrow$ negative eigenvalues in lambda
		\IF{$\text{valid\_lambda}$ is empty}
		\STATE $\text{sigma\_max} \leftarrow 1$
		\ELSE
		\STATE $\text{sigma\_max} \leftarrow 1 / |\min(\text{valid\_lambda})|$
		\ENDIF
		\STATE $\text{sigma\_values} \leftarrow \text{sigma\_max} \times \text{linspace}(0.01, 0.99, 60)$ \COMMENT{relative grid: $\sigma = x\,\sigma^*$ with $x$ equidistant in $[0.01, 0.99]$}

		\STATE \textit{For each $\sigma$, compute DomiRank and entropy:}
		\FOR{$i = 1$ \TO $\text{length}(\text{sigma\_values})$}
		\STATE $\sigma \leftarrow \text{sigma\_values}[i]$
		\STATE $\boldsymbol{\Gamma} \leftarrow (\sigma A + I)^{-1} (\sigma A \mathbf{1})$
		\STATE $\boldsymbol{\Gamma} \leftarrow \mathrm{Re}(\boldsymbol{\Gamma})$
		\STATE $\boldsymbol{\Gamma}_+ \leftarrow \max(\boldsymbol{\Gamma}, 0)$ \COMMENT{effective distribution: negative values are truncated, then renormalized (see Section 6.3)}
		\STATE $\boldsymbol{P} \leftarrow \boldsymbol{\Gamma}_+ / \sum(\boldsymbol{\Gamma}_+)$
		\STATE $\text{entropy}[i] \leftarrow -\sum(\boldsymbol{P} \cdot \log_2 \boldsymbol{P})$ \COMMENT{entropy of the effective distribution}
		\ENDFOR

		\STATE\textit{Post-processing (not part of the core computation):} plot entropy and DomiRank distributions against $\sigma$; compare with the principal eigenvector; compute the rank correlations and positive fractions of Figures~\ref{figure13} and~\ref{figure14} from the vectors already obtained above
	\end{algorithmic}
\end{algorithm}

Figure~\ref{figure4} shows the DomiRank distributions of the 11-vertex star graph, constructed by the above algorithm, under different competition parameters $\sigma$; the three panels represent low, medium, and high competition intensities. As $\sigma$ increases, DomiRank shifts from identifying ``important nodes'' to identifying ``dominant key nodes.''

\begin{figure}[!h]
	\centering
	\includegraphics[width=1\textwidth]{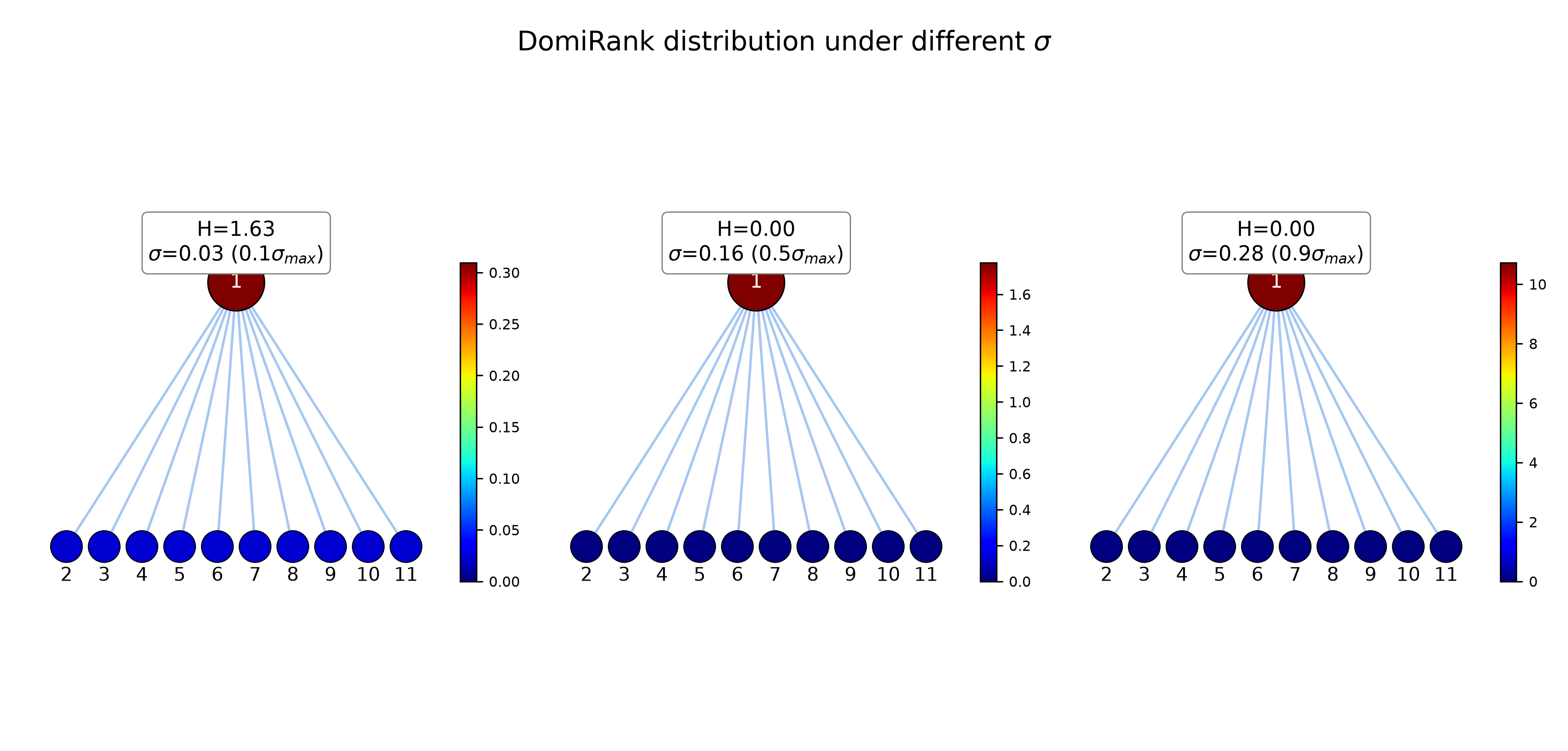}
	\caption{DomiRank distributions of the 11-vertex star graph under low, medium, and high competition intensities (three panels, all within the admissible range $(0, \sigma^*)$; for this graph $\sigma^* = 1/\sqrt{10} \approx 0.3162$).}
	\textbf{Alt text:}\ Three bar charts showing the DomiRank distribution of the 11-vertex star graph at low, medium and high competition intensities. As the intensity rises, the values of the nine leaves shrink towards zero and the value of the hub grows sharply.
	\label{figure4}
\end{figure}

Figure~\ref{figure5} shows how the competition intensity affects the uniformity of the node-importance distribution. Here the entropy $H$ measures the uncertainty of the DomiRank distribution. As shown, the entropy decreases monotonically with the competition parameter $\sigma$, consistent with Theorem~\ref{thm:entropy_decrease} and Conjecture~\ref{conj:entropy}.

\begin{figure}[!h]
	\centering
	\includegraphics{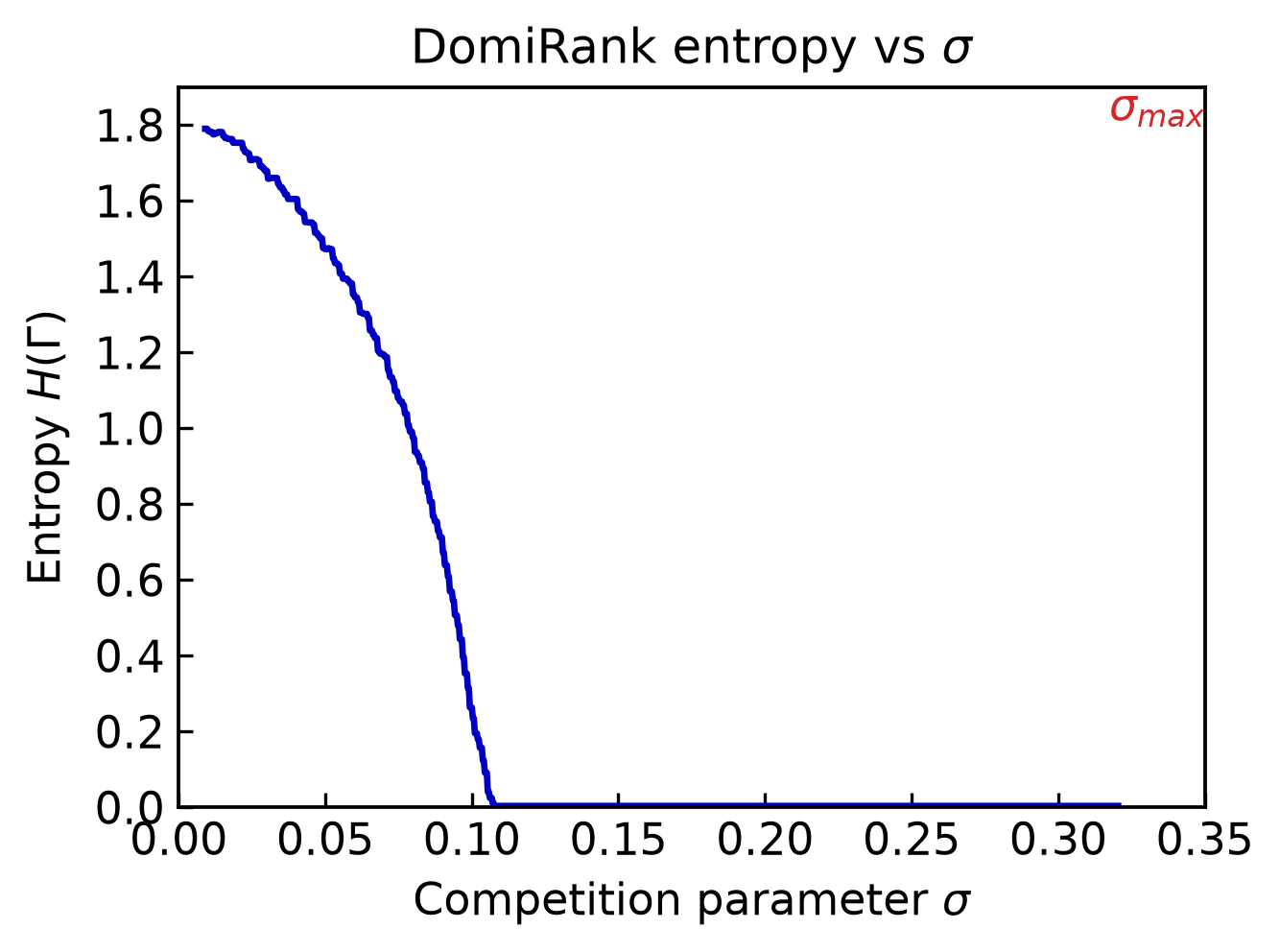}\\
	\caption{DomiRank entropy of the star graph versus $\sigma$.}
	\textbf{Alt text:}\ A line plot of the DomiRank entropy of the star graph against the relative competition intensity x from 0.01 to 0.99. The curve decreases monotonically from about 3.4 bits to about 0.5 bits.
	\label{figure5}
\end{figure}

Figure~\ref{figure6} highlights the uniqueness of the DomiRank mechanism by comparing it with the traditional principal-eigenvector centrality. DomiRank does not find the ``most popular nodes'' (eigenvector centrality) but the ``most dominant nodes.''

\begin{figure}[!h]
	\centering
	\includegraphics{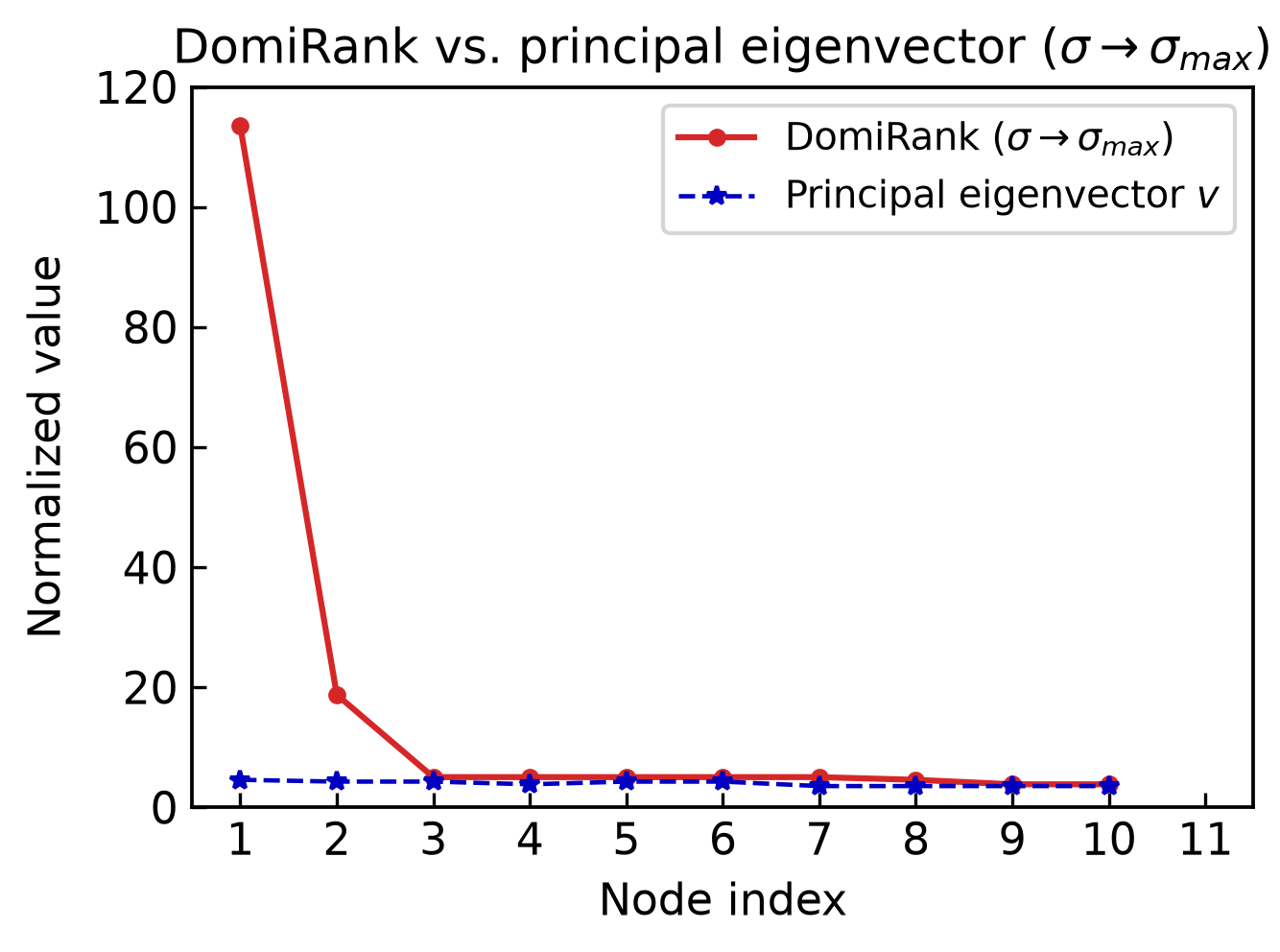}\\
	\caption{Comparison between DomiRank and principal-eigenvector centrality on the star graph (admissible range $(0, \sigma^*)$ with $\sigma^* = 1/\sqrt{10} \approx 0.3162$): DomiRank does not identify the most popular nodes but the most dominant ones.}
	\textbf{Alt text:}\ A line plot of the DomiRank entropy against the relative competition intensity for a second benchmark network, again decreasing monotonically across the grid.
	\label{figure6}
\end{figure}

\begin{conjecture}\label{conj:star}
Among connected graphs \emph{on $N$ vertices} (with $N$ fixed), the star graph $S_N$ is the graph of minimum DomiRank entropy. Since the attainable range of $H$ is $[0, \log_2 N]$, the comparison is meaningful only at fixed $N$; equivalently, $H(G) \ge H(S_N)$ whenever $|V(G)| = |V(S_N)|$.
\end{conjecture}

\begin{remark}[Numerical evidence for the conjecture]\label{rem:star_evidence}
We provide systematic numerical support for Conjecture~\ref{conj:star} (under the following operationalization: for any connected graph $G$ on $N$ vertices and any relative competition intensity $x \in (0, 1)$, compare the normalized entropy of $G$ at $\sigma = x\sigma_G^*$ with that of the star $S_N$ at $\sigma = x\sigma_{S_N}^*$). For all connected graphs with $N \le 7$ (996 graphs in total --- connected, simple, undirected and unlabeled, of which 853 have $N = 7$), we computed at the 19 sampling points $x \in \{0.05, 0.10, \dots, 0.95\}$: in all 18,924 comparisons, $H_G \ge H_{S_N}$ held, with no counterexample; we further scanned 900 random networks with $N \in \{10, 20, 50\}$ (300 each of Erd\H{o}s--R\'enyi, Barab\'asi--Albert, and Watts--Strogatz models), again finding no counterexample.
\end{remark}

\section{DomiRank Properties under Graph\\allowbreak Transformations}

The previous two sections characterized DomiRank from the perspectives of symmetry and distributional measures. This section further examines the influence of structural interventions on the DomiRank vector, focusing on typical graph mappings and transformations (vertex similarity, vertex partitions, edge swaps, and $m$-products). Four results are proved here: Theorem~\ref{thm:similar_neighbor} on vertex similarity, Theorem~\ref{thm:partition} on equitable partitions, Theorem~\ref{thm:edge_swap} on edge swaps, and Theorem~\ref{thm:m_product} on $m$-products. The remaining transformation, edge deletion, is the most delicate of the four: Theorem~\ref{thm:edge_deletion} gives an identity that holds under mild conditions together with a necessary-and-sufficient condition for the stronger exact scaling, and Remark~\ref{rem:edge_deletion_sharp} records a four-cycle showing that the stronger statement is false in general.

\begin{definition}\label{def:similar}
Let $G$ be a graph and $\mathrm{Aut}(G)$ its automorphism group. If there exists an automorphism $\psi \in \mathrm{Aut}(G)$ with $\psi(u) = v$, then vertices $u$ and $v$ are called \textbf{similar}. By Theorem~\ref{thm:automorphism}, similar vertices have equal DomiRank values.
\end{definition}

\begin{theorem}\label{thm:similar_neighbor}
Let $u$ and $v$ be two vertices of a graph $G$ with equal degrees. If every vertex in $N(u)\setminus N(v)$ is similar to $u$, and every vertex in $N(v)\setminus N(u)$ is similar to $v$, then $\Gamma_u = \Gamma_v$.
\end{theorem}

\begin{proof}
If $N(v) = N(u)$, Lemma~\ref{lemma:same_neighbor} gives $\Gamma_u = \Gamma_v$.

Assume $N(v) \neq N(u)$. For $w \in N(u)\setminus N(v)$, since every vertex in $N(u)\setminus N(v)$ is similar to $u$, we have $\Gamma_u = \Gamma_w$. By equation \eqref{eq:node_solution},
\[
\begin{aligned}
	\Gamma_u &= \sigma \sum_{k \in N(u)} \left[ \theta - \Gamma_k \right] \\
	&= \sigma \sum_{p \in N(u)\cap N(v)} \left[ \theta - \Gamma_p \right] + \sigma \sum_{w \in N(u)\setminus N(v)} \left[ \theta - \Gamma_w \right] \\
	&= \sigma d(u)\theta - \sigma \sum_{p \in N(u)\cap N(v)} \Gamma_p - \sigma \sum_{w \in N(u)\setminus N(v)} \Gamma_w.
\end{aligned}
\]

Since every $w \in N(u)\setminus N(v)$ is similar to $u$, $\Gamma_w = \Gamma_u$, so
\begin{align}\label{eq:similar_u}
\left[1 + \left(d(u) - |N(u)\cap N(v)|\right)\sigma\right]\Gamma_u = \sigma d(u)\theta - \sigma \sum_{p \in N(u)\cap N(v)} \Gamma_p.
\end{align}

Similarly, for $v$,
\begin{align}\label{eq:similar_v}
\left[1 + \left(d(v) - |N(u)\cap N(v)|\right)\sigma\right]\Gamma_v = \sigma d(v)\theta - \sigma \sum_{p \in N(u)\cap N(v)} \Gamma_p.
\end{align}

Since $d(u) = d(v)$, the right-hand sides of \eqref{eq:similar_u} and \eqref{eq:similar_v} are identical and the left-hand coefficients are identical; hence $\Gamma_u = \Gamma_v$.
\end{proof}

\begin{remark}
The hypotheses of Theorem~\ref{thm:similar_neighbor} are weaker than similarity: they do not require $u$ and $v$ to lie in the same orbit. In the graph of Figure~\ref{figure7}, $u$ and $v$ share the common neighbors $a, b$; the other neighbor $c$ of $u$ is similar to $u$, and the other neighbor $y$ of $v$ is similar to $v$. For $\sigma = 0.3, \theta = 1$, the DomiRank vector of this graph is $\left[ \dfrac{24}{29},\; \dfrac{9}{29},\; \dfrac{9}{29},\; \dfrac{24}{29},\; \dfrac{9}{29},\; \dfrac{9}{29} \right]$, so $u$ and $v$ attain equal DomiRank values; in this particular instance $u$ and $v$ are in fact similar (the automorphism $\psi: a \mapsto a,\; b \mapsto b,\; u \mapsto v,\; c \mapsto y$ maps the graph to itself), while the second instance of Example~\ref{ex:orbit} (Figure~\ref{figure15}) provides a case in which vertices from different orbits still share equal DomiRank values.
\end{remark}

\begin{figure}[!h]
	\centering
	\includegraphics[width=0.95\textwidth]{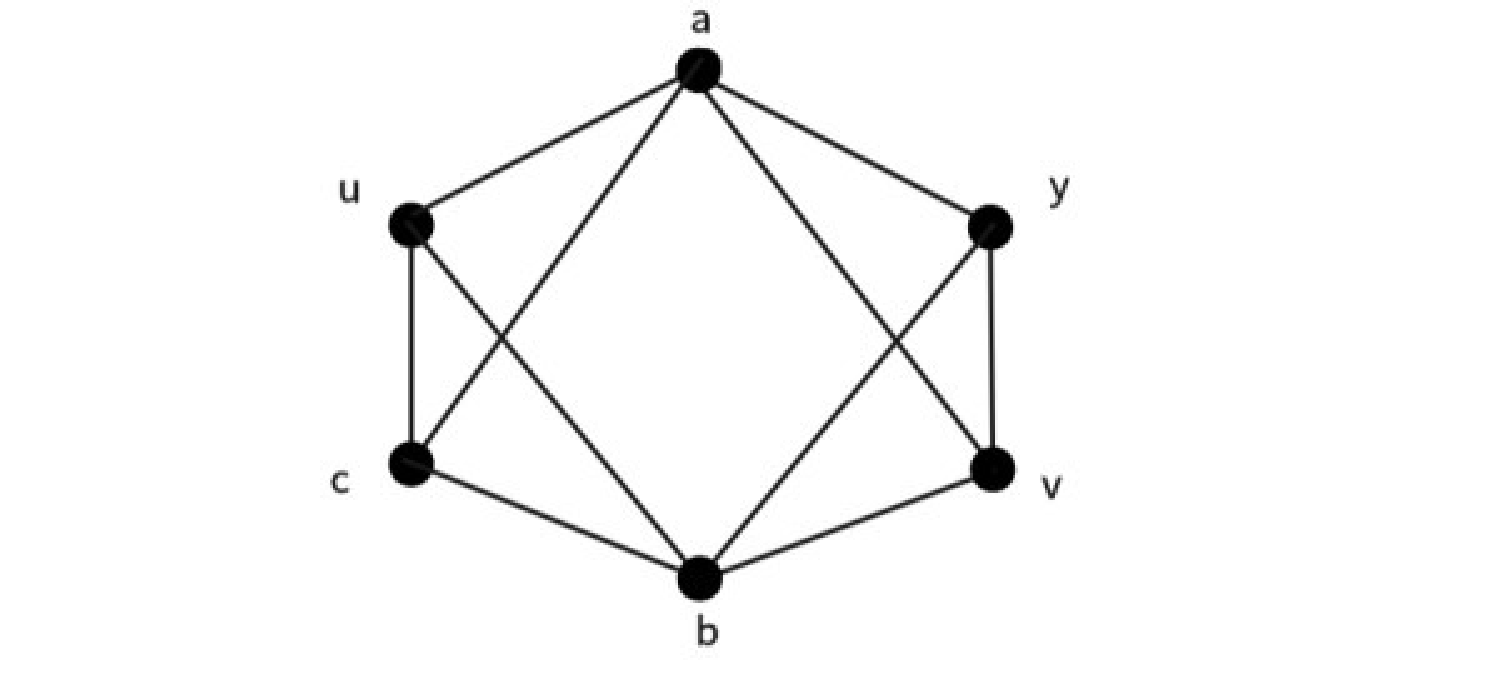}\\
	\caption{The graph of the remark after Theorem~\ref{thm:similar_neighbor}: vertices $u$ and $v$ attain equal DomiRank values.}
	\textbf{Alt text:}\ A drawing of two vertices u and v with equal DomiRank values that are not related by any automorphism, used as the sharpness example after Theorem 5.2.
	\label{figure7}
\end{figure}

\begin{definition}\label{def:partition}
Let the vertex set of a graph $G$ have a partition $V(G) = V_1 \cup V_2 \cup \dots \cup V_m$. If for any $i, j$, every vertex in $V_i$ is adjacent to the same number of vertices in $V_j$, the partition is called an equitable partition. Let $d_{i,j}$ denote the number of vertices in $V_j$ adjacent to a vertex in $V_i$.
\end{definition}

\begin{theorem}\label{thm:partition}
Let $G$ be a graph of order $N$ whose vertex set has an equitable partition $V(G)=V_1\cup V_2\cup\dots\cup V_m$. Then all vertices in the same $V_i$ have equal DomiRank values.
\end{theorem}

\begin{proof}
The DomiRank equation $(\sigma A + I)\boldsymbol{\Gamma} = \sigma \theta A \mathbf{1}_{N \times 1}$ has a unique solution for $\sigma \in (0, -1/\lambda_N)$ (since $\sigma A + I$ is invertible).

Assume the DomiRank vector $\boldsymbol{\Gamma}$ takes a constant value $r_i$ inside each block, i.e., $\Gamma_v = r_i$ for every $v \in V_i$. For $v \in V_i$, the left-hand side is
\[
\left[(\sigma A + I)\boldsymbol{\Gamma}\right]_v = \Gamma_v + \sigma \sum_{u} A_{v,u} \Gamma_u = r_i + \sigma \sum_j d_{i,j} r_j.
\]

The right-hand side is
\[
\left[\sigma \theta A \mathbf{1}\right]_v = \sigma \theta \sum_u A_{v,u} = \sigma \theta S_i,
\]
where $S_i = \sum_j d_{i,j}$ is the degree of a vertex in $V_i$.

Therefore, $r_i$ satisfies the system
\begin{align}\label{eq:partition_system}
r_i + \sigma \sum_j d_{i,j} r_j = \sigma \theta S_i, \quad i = 1, 2, \dots, m,
\end{align}
a linear system $(I + \sigma D)\boldsymbol{r} = \sigma \theta \boldsymbol{s}$ in $r_1, r_2, \dots, r_m$, where $D_{i,j} = d_{i,j}$ and $\boldsymbol{s} = (S_1, \dots, S_m)^T$.

The key observation is that this ansatz is consistent: for every vertex $v \in V_i$, the equation reduces to the same equation (depending only on $i$, not on the specific choice of $v$). Since the original system $(\sigma A + I)\boldsymbol{\Gamma} = \sigma \theta A \mathbf{1}$ has a unique solution, and the constant-value ansatz satisfies the equations of all vertices, uniqueness implies that the solution takes constant values inside each block.
\end{proof}

\begin{example}\label{ex:bipartite}
Consider the complete bipartite graph $K_{2,3}$ of Figure~\ref{figure8}, with $V_1 = \{a, b\}$ and $V_2 = \{m, n, v\}$. At $\sigma = 0.3$, by formula \eqref{eq:solution}, every vertex in $V_1$ has DomiRank value $\dfrac{18}{23}$ and every vertex in $V_2$ has DomiRank value $\dfrac{3}{23}$, verifying Theorem~\ref{thm:partition}.
\end{example}

\begin{figure}[!h]
	\centering
	\includegraphics[width=0.72\textwidth]{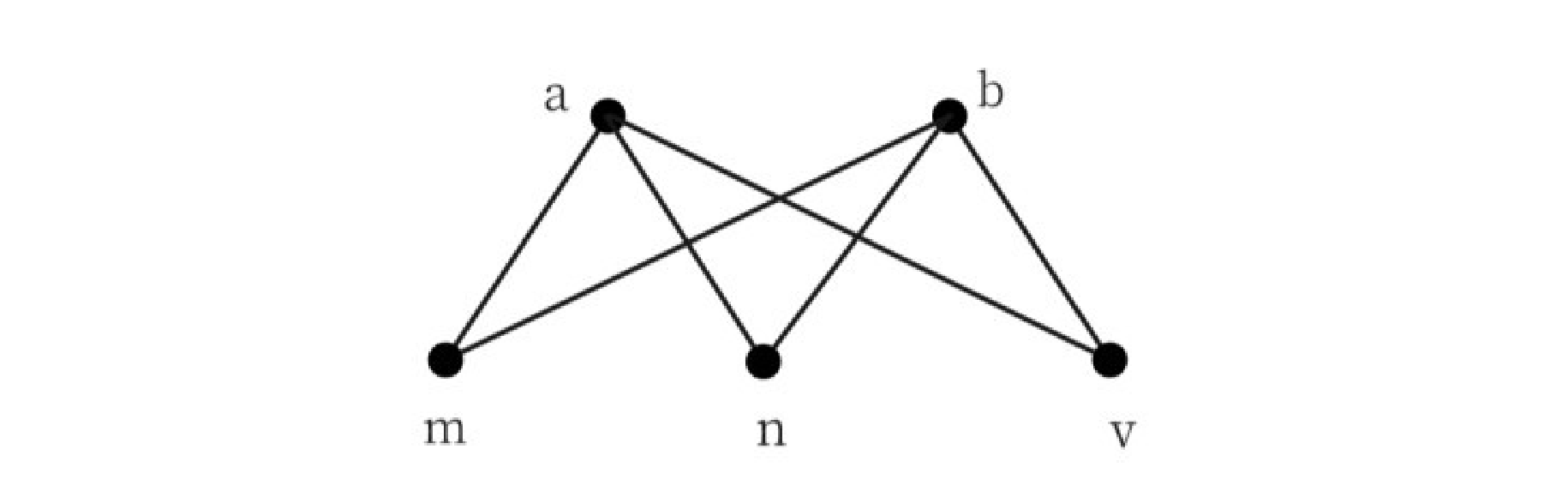}\\
	\caption{The complete bipartite graph $K_{2,3}$ of Example~\ref{ex:bipartite}, used there to verify the equitable-partition statement of Theorem~\ref{thm:partition}: the two vertices of $V_1=\{a,b\}$ receive $18/23$ and the three vertices of $V_2=\{m,n,v\}$ receive $3/23$ at $\sigma=0.3$.}
	\textbf{Alt text:}\ A drawing of the complete bipartite graph K2,3 with part V1 of two vertices a and b and part V2 of three vertices m, n and v. It is the graph of Example 5.6, used to verify that vertices in the same block of an equitable partition receive equal DomiRank values.
	\label{figure8}
\end{figure}

\begin{example}\label{ex:general_partition}
As shown in Figure~\ref{figure9}, partition 7 vertices into $V = V_1 \cup V_2 \cup V_3$ with $V_1 = \{1\}$, $V_2 = \{2, 3\}$, and $V_3 = \{4, 5, 6, 7\}$. This partition satisfies the equitability condition, and at $\sigma = 0.3$ the DomiRank vector of the graph is
\[
\left[ \frac{252}{247},\; \frac{72}{247},\; \frac{72}{247},\; \frac{249}{494},\; \frac{249}{494},\; \frac{249}{494},\; \frac{249}{494} \right].
\]
Vertices in the same $V_i$ have equal DomiRank values, verifying Theorem~\ref{thm:partition}.
\end{example}

\begin{figure}[!h]
	\centering
	\includegraphics[width=0.82\textwidth]{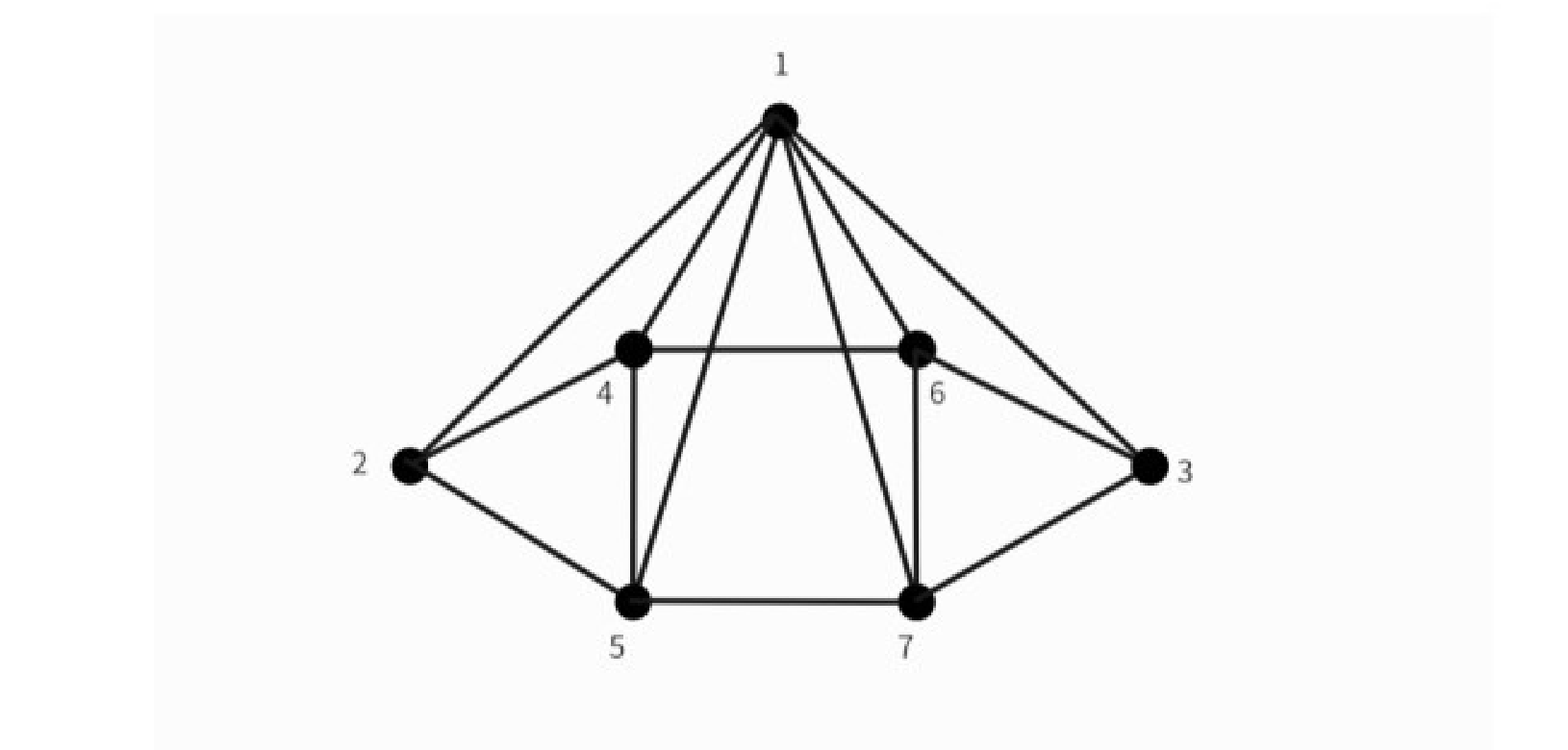}\\
	\caption{An equitable partition $V_1=\{1\}$, $V_2=\{2,3\}$, $V_3=\{4,5,6,7\}$ of the graph of Example~\ref{ex:general_partition}.}
	\textbf{Alt text:}\ A drawing of a seven-vertex graph with the equitable partition V1={1}, V2={2,3} and V3={4,5,6,7}. Vertices 2 and 3, and vertices 4 to 7, are drawn symmetrically within their blocks.
	\label{figure9}
\end{figure}

\begin{theorem}\label{thm:edge_deletion}
Let $F$ be a subset of $E(G)$ such that $|E(G)| = k|F|$ for some integer $k \ge 2$, and let $H$ be the spanning subgraph of $G$ obtained by deleting all edges not in $F$. Suppose that for every $x \in V(G)$ the following two conditions hold:
\begin{enumerate}
	\item[(i)] $d_G(x) = k\, d_H(x)$;
	\item[(ii)] $N_G(x)$ can be partitioned as $N_G(x) = Y_1 \cup Y_2 \cup \cdots \cup Y_k$ with $Y_1 = N_H(x)$, and for each $i = 1, 2, \dots, k$ there exists a bijection $\varphi_i$ from $Y_1$ to $Y_i$ such that every $v \in Y_1$ and $\varphi_i(v)$ are similar vertices.
\end{enumerate}
Then the DomiRank vectors $\boldsymbol{\Gamma}^G$ of $G$ and $\boldsymbol{\Gamma}^H$ of $H$ satisfy, for every $x \in V(G)$,
\begin{align}\label{eq:edge_deletion_identity}
\Gamma^G_x = k \left[ \sigma\theta\, d_H(x) - \sigma \sum_{y \in N_H(x)} \Gamma^G_y \right],
\end{align}
and the exact scaling $\boldsymbol{\Gamma}^G = k\, \boldsymbol{\Gamma}^H$ holds if and only if
\begin{align}\label{eq:edge_deletion_condition}
\sum_{y \in N_H(x)} \Gamma^G_y = \sum_{y \in N_H(x)} \Gamma^H_y \qquad \text{for every } x \in V(G).
\end{align}
\end{theorem}

\begin{proof}
Let $\boldsymbol{\Gamma}^G$ be the DomiRank vector of $G$. For every $x \in V(G)$,
\begin{align*}
\Gamma^G_x &= \sigma \sum_{y \in N_G(x)} \left[ \theta - \Gamma^G_y \right] \\
&= \sigma \sum_{i=1}^{k} \sum_{z \in Y_i} \left[ \theta - \Gamma^G_z \right] \\
&= \sigma\theta\, d_G(x) - \sigma \sum_{i=1}^{k} \sum_{z \in Y_i} \Gamma^G_z.
\end{align*}
By condition (ii), for each $z \in Y_i$ there exists a unique $y \in Y_1 = N_H(x)$ such that $z$ and $y$ are similar; since similarity is membership of a common orbit of $\mathrm{Aut}(G)$ (Definition~\ref{def:similar}), this gives $\Gamma^G_z = \Gamma^G_y$ by Theorem~\ref{thm:automorphism}. Together with condition (i), $d_G(x) = k\, d_H(x)$, this gives
\[
\Gamma^G_x = \sigma\theta\, k\, d_H(x) - k\sigma \sum_{y \in N_H(x)} \Gamma^G_y
= k \left[ \sigma\theta\, d_H(x) - \sigma \sum_{y \in N_H(x)} \Gamma^G_y \right],
\]
which is \eqref{eq:edge_deletion_identity}. For the second assertion, recall the defining equation of $H$, $\Gamma^H_x = \sigma\theta\, d_H(x) - \sigma \sum_{y \in N_H(x)} \Gamma^H_y$. If $\boldsymbol{\Gamma}^G = k \boldsymbol{\Gamma}^H$, substituting into \eqref{eq:edge_deletion_identity} and comparing with the defining equation of $H$ yields \eqref{eq:edge_deletion_condition}. Conversely, if \eqref{eq:edge_deletion_condition} holds, then \eqref{eq:edge_deletion_identity} gives
\[
\Gamma^G_x = k \left[ \sigma\theta\, d_H(x) - \sigma \sum_{y \in N_H(x)} \Gamma^H_y \right] = k\, \Gamma^H_x
\]
for every $x$.
\end{proof}

\begin{remark}[The exact scaling fails in general]\label{rem:edge_deletion_sharp}
Conditions (i)--(ii) alone do not imply $\boldsymbol{\Gamma}^G = k \boldsymbol{\Gamma}^H$: DomiRank values are globally coupled, so $\Gamma^G_y \ne \Gamma^H_y$ in general and condition \eqref{eq:edge_deletion_condition} is genuinely restrictive. The smallest counterexample is $G = C_4$ with $F$ a perfect matching ($k = 2$): conditions (i)--(ii) hold (for every vertex its two neighbors in $G$ have the same neighborhood and are therefore similar), yet $C_4$ is $2$-regular and $H = 2K_2$ is $1$-regular, so $\Gamma^G_x = \dfrac{2\sigma\theta}{1 + 2\sigma}$ while $k\, \Gamma^H_x = \dfrac{2\sigma\theta}{1 + \sigma}$, and the two agree only at $\sigma = 0$ (for instance, at $\sigma = 0.25$, $\Gamma^G_x = 0.3333 \ne 0.4000 = k\, \Gamma^H_x$). Accordingly \eqref{eq:edge_deletion_condition} fails, since $\Gamma^G_y \ne \Gamma^H_y$; the identity \eqref{eq:edge_deletion_identity}, which holds under (i)--(ii) alone, quantifies the deviation from exact scaling.
\end{remark}

\begin{theorem}\label{thm:edge_swap}
Let $G$ be a graph with two edges $ua$ and $vb$, and let $G'$ be obtained from $G$ by replacing $ua$ with $ub$ and $vb$ with $va$, as in Figure~\ref{figure10}. If $u$ and $v$ are similar, $a$ and $b$ are similar, and $ub$, $va$ are not edges of $G$, then $G'$ has the same DomiRank vector as $G$.
\end{theorem}
\begin{proof}
Since $u$ and $v$ are similar and $a$ and $b$ are similar, Theorem~\ref{thm:automorphism} gives $\Gamma_u = \Gamma_v$, $\Gamma_a = \Gamma_b$, and $d(u) = d(v)$, $d(a) = d(b)$.

Partition $V(G) = V(G') = V_1 \cup V_2$ with $V_1 = \{u, v, a, b\}$ and $V_2$ its complement. Let the adjacency matrices of $G$ and $G'$ be $P$ and $P'$, in block form
\[
P = \begin{bmatrix} P_1 & P_2 \\ P_3 & P_4 \end{bmatrix},
\qquad
P' = \begin{bmatrix} P_1' & P_2 \\ P_3 & P_4 \end{bmatrix},
\]
where the sub-blocks corresponding to $V_1$ are
\[
P_1 = \begin{bmatrix}
0 & 1 & 1 & 0 \\
1 & 0 & 0 & 1 \\
1 & 0 & 0 & 1 \\
0 & 1 & 1 & 0
\end{bmatrix},
\qquad
P_1' = \begin{bmatrix}
0 & 1 & 0 & 1 \\
1 & 0 & 1 & 0 \\
0 & 1 & 0 & 1 \\
1 & 0 & 1 & 0
\end{bmatrix}.
\]

Split $\boldsymbol{\Gamma}$ according to the partition of $G$: $\boldsymbol{\Gamma} = \begin{bmatrix} \boldsymbol{\Gamma}_1 \\ \boldsymbol{\Gamma}_2 \end{bmatrix}$, where $\boldsymbol{\Gamma}_1 = \begin{bmatrix} \Gamma_u & \Gamma_v & \Gamma_a & \Gamma_b \end{bmatrix}^T$, and split $\mathbf{1} = \begin{bmatrix} \mathbf{1}_1 \\ \mathbf{1}_2 \end{bmatrix}$.

The block form of $(\sigma P + I)\boldsymbol{\Gamma} = \sigma \theta P \mathbf{1}$ is
\begin{align}\label{eq:block}
\sigma \begin{bmatrix} P_1 & P_2 \\ P_3 & P_4 \end{bmatrix} \begin{bmatrix} \boldsymbol{\Gamma}_1 \\ \boldsymbol{\Gamma}_2 \end{bmatrix} + \begin{bmatrix} \boldsymbol{\Gamma}_1 \\ \boldsymbol{\Gamma}_2 \end{bmatrix} = \sigma \theta \begin{bmatrix} P_1 & P_2 \\ P_3 & P_4 \end{bmatrix} \begin{bmatrix} \mathbf{1}_1 \\ \mathbf{1}_2 \end{bmatrix}.
\end{align}

Since $\Gamma_u = \Gamma_v$ and $\Gamma_a = \Gamma_b$, one verifies directly that $P_1 \boldsymbol{\Gamma}_1 = P_1' \boldsymbol{\Gamma}_1$ and $P_1 \mathbf{1}_1 = P_1' \mathbf{1}_1$. Substituting into \eqref{eq:block},
\[
\sigma \begin{bmatrix} P_1' & P_2 \\ P_3 & P_4 \end{bmatrix} \begin{bmatrix} \boldsymbol{\Gamma}_1 \\ \boldsymbol{\Gamma}_2 \end{bmatrix} + \begin{bmatrix} \boldsymbol{\Gamma}_1 \\ \boldsymbol{\Gamma}_2 \end{bmatrix} = \sigma \theta \begin{bmatrix} P_1' & P_2 \\ P_3 & P_4 \end{bmatrix} \begin{bmatrix} \mathbf{1}_1 \\ \mathbf{1}_2 \end{bmatrix},
\]
i.e., $(\sigma P' + I)\boldsymbol{\Gamma} = \sigma \theta P' \mathbf{1}$. By the uniqueness of the DomiRank vector, $\boldsymbol{\Gamma}$ is also the DomiRank vector of $G'$.
\end{proof}

\begin{figure}[!h]
	\centering
	\includegraphics[width=1\textwidth]{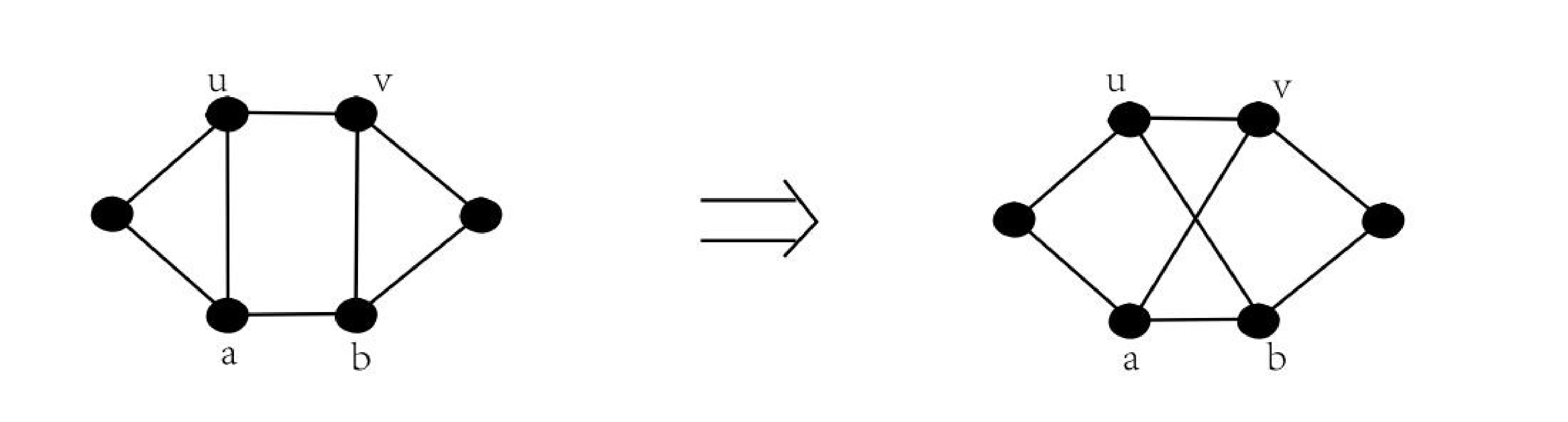}
	\caption{Edge swaps preserve the DomiRank vector (Theorem~\ref{thm:edge_swap}); $G'$ is obtained from $G$ by replacing the edges $ua$ and $vb$ with $ub$ and $va$. The statement holds for every admissible $\sigma \in (0, -1/\lambda_N)$.}
	\textbf{Alt text:}\ Two drawings of the same six-vertex graph before and after an edge swap that replaces edges ua and vb with ub and va. The two graphs have identical DomiRank vectors, which is the content of Theorem 5.10.
	\label{figure10}
\end{figure}

\begin{definition}\label{def:m_product}
The $m$-product of a graph $G$, denoted by $G^{(m)}$, is obtained by replacing each vertex $v_i$ of $G$ with an independent set $\{v_{i,1}, v_{i,2}, \dots, v_{i,m}\}$ of $m$ vertices, where two vertices $v_{i,j}$ and $v_{k,l}$ are adjacent in $G^{(m)}$ if and only if $v_i$ and $v_k$ are adjacent in $G$.
\end{definition}

\begin{theorem}\label{thm:m_product}
Let $G$ be a graph with DomiRank vector $\boldsymbol{\Gamma}$. Then the DomiRank vector $\boldsymbol{\Gamma}^{(m)}$ of $G^{(m)}$ satisfies: for any $1 \le i \le n$ and $1 \le j \le m$,
\[
\Gamma_{v_{i,j}}^{(m)} = \Gamma_i^{(\sigma m)},
\]
where $\Gamma_i^{(\sigma m)}$ is the DomiRank value of $G$ at parameter $\sigma m$, and the statement holds for every $\sigma \in (0, -1/(m\lambda_N))$, the admissible range of $G^{(m)}$: since $\lambda_N(A^{(m)}) = m\,\lambda_N(A)$, this is exactly the range on which both sides are defined, and $\sigma m$ then lies in the admissible range $(0, -1/\lambda_N)$ of $G$.
\end{theorem}

\begin{proof}
By Definition~\ref{def:m_product}, the adjacency matrix of $G^{(m)}$ is
\begin{align}\label{eq:kron}
A^{(m)} = A \otimes J_m,
\end{align}
where $\otimes$ denotes the Kronecker product and $J_m \in \mathbb{R}^{m \times m}$ is the all-ones matrix. That is, $A^{(m)}$ is an $mn \times mn$ matrix whose $(i, k)$-block is $A_{ik} \cdot J_m$.

By equation \eqref{eq:solution}, the DomiRank vector of $G^{(m)}$ is
\[
\boldsymbol{\Gamma}^{(m)} = \theta \sigma \left(\sigma A^{(m)} + I_{mn}\right)^{-1} A^{(m)} \mathbf{1}_{mn}.
\]

Since $J_m \mathbf{1}_m = m \mathbf{1}_m$,
\[
A^{(m)} \mathbf{1}_{mn} = \left(A \otimes J_m\right)\left(\mathbf{1}_n \otimes \mathbf{1}_m\right) = \left(A \mathbf{1}_n\right) \otimes \left(m \mathbf{1}_m\right).
\]

Let $U$ be an orthogonal matrix such that $J_m = U \Lambda U^T$, where $\Lambda$ is diagonal with eigenvalues $m$ (once) and $0$ ($m-1$ times). Then
\[
\begin{aligned}
	\sigma A^{(m)} + I_{mn} &= \sigma \left(A \otimes J_m\right) + I_n \otimes I_m \\
	&= \left(I_n \otimes U\right) \left[\sigma \left(A \otimes \Lambda\right) + I_n \otimes I_m\right] \left(I_n \otimes U^T\right).
\end{aligned}
\]

The matrix $\sigma(A \otimes \Lambda) + I_n \otimes I_m$ is block diagonal. Since $\Lambda$ has only one nonzero eigenvalue $m$, its inverse contributes only on the block corresponding to eigenvalue $m$, which is $I_n + \sigma m A$; the remaining blocks are $I_n$. Hence
\[
\left[\sigma(A \otimes \Lambda) + I_n \otimes I_m\right]^{-1} A^{(m)} \mathbf{1}_{mn} = m(I_n + \sigma m A)^{-1} A \mathbf{1}_n \otimes \mathbf{1}_m.
\]

Therefore
\begin{align}\label{eq:m_product_result}
\boldsymbol{\Gamma}^{(m)} = \theta \sigma m (I_n + \sigma m A)^{-1} A \mathbf{1}_n \otimes \mathbf{1}_m.
\end{align}

Let $\boldsymbol{u} = \theta \sigma m (I_n + \sigma m A)^{-1} A \mathbf{1}_n$; then $\boldsymbol{u} \otimes \mathbf{1}_m$ has the structure
\[
\boldsymbol{u} \otimes \mathbf{1}_m = \begin{bmatrix}
u_1 \mathbf{1}_m \\
u_2 \mathbf{1}_m \\
\vdots \\
u_n \mathbf{1}_m
\end{bmatrix}.
\]

That is, for each group $i$, all vertices $v_{i,j}$ ($j = 1, \dots, m$) share the same DomiRank value
\[
\Gamma_{v_{i,j}}^{(m)} = u_i = \left[ \theta \sigma m (I_n + \sigma m A)^{-1} A \mathbf{1}_n \right]_i.
\]

But $\theta \sigma m (I_n + \sigma m A)^{-1} A \mathbf{1}_n$ is precisely the DomiRank vector $\boldsymbol{\Gamma}^{(\sigma m)}$ of $G$ at competition parameter $\sigma m$, so
\[
\Gamma_{v_{i,j}}^{(m)} = \Gamma_i^{(\sigma m)}.
\]
\end{proof}

\begin{theorem}\label{thm:local_regular}
Let $G$ be a graph of order $n$, $x$ a vertex, and $\boldsymbol{\Gamma}$ the DomiRank vector of $G$. Take $\theta = 1$ and assume $d(x) = d(y) = k$ for all $y \in N(x)$.

 \begin{enumerate}
	\item[(i)] If $\Gamma_y = \dfrac{\sigma k}{\sigma k + 1}$ for all $y \in N(x)$, then $\Gamma_x = \dfrac{\sigma k}{\sigma k + 1}$.
	\item[(ii)] Given $y \in N(x)$, if $\Gamma_z = \Gamma_x = \dfrac{\sigma k}{\sigma k + 1}$ for all $z \in N(x)\setminus\{y\}$, then $\Gamma_y = \dfrac{\sigma k}{\sigma k + 1}$.
	\item[(iii)] If $\Gamma_x = \dfrac{\sigma k}{\sigma k + 1}$ and all vertices in $N(x)$ have equal DomiRank values, then $\Gamma_y = \dfrac{\sigma k}{\sigma k + 1}$ for every $y \in N(x)$.
 \end{enumerate}
\end{theorem}

\begin{proof}
\begin{enumerate}
 \item[(i)] By equation \eqref{eq:node_solution} (with $\theta = 1$),
\begin{align*}
\Gamma_x &= \sigma k - \sigma \sum_{j \in N(x)} \Gamma_j \\
&= \sigma k - \sigma k \cdot \frac{\sigma k}{\sigma k + 1}
 = \sigma k \left(1 - \frac{\sigma k}{\sigma k + 1}\right)
 = \frac{\sigma k}{\sigma k + 1}.
\end{align*}

\item[(ii)]
\[
\frac{\sigma k}{\sigma k + 1} = \Gamma_x = \sigma k - \sigma (k-1) \frac{\sigma k}{\sigma k + 1} - \sigma \Gamma_y.
\]
Solving,
\[
\sigma \Gamma_y = \sigma k - \frac{\sigma k}{\sigma k + 1} - \sigma(k-1)\frac{\sigma k}{\sigma k + 1} = \sigma k - \frac{\sigma k(1 + \sigma k - \sigma)}{\sigma k + 1},
\]
\[
\Gamma_y = k - \frac{k(1 + \sigma k - \sigma)}{\sigma k + 1} = \frac{k(\sigma k + 1) - k(1 + \sigma k - \sigma)}{\sigma k + 1} = \frac{\sigma k}{\sigma k + 1}.
\]

\item[(iii)]
\[
\frac{\sigma k}{\sigma k + 1} = \Gamma_x = \sigma k - \sigma k \Gamma_y,
\]
\[
\sigma k \Gamma_y = \sigma k - \frac{\sigma k}{\sigma k + 1} = \frac{\sigma^2 k^2}{\sigma k + 1},
\]
\[
\Gamma_y = \frac{\sigma k}{\sigma k + 1}.
\]
\end{enumerate}
\end{proof}

\section{Specificity of the Parameter $\sigma$ and Experiments on Real Networks}
\label{sec:sigma}

The previous sections established the theoretical framework of DomiRank from the perspectives of symmetry, distributional measures, and structural transformations. As the only tunable parameter of the model, the competition parameter $\sigma$ determines to what extent the evaluation results change and according to what laws --- questions that decide the practical value of DomiRank. This section first gives an analytic characterization of the sensitivity and the limiting behavior of $\sigma$, and then conducts systematic experiments on four real-world networks to verify and exhibit the specificity of $\sigma$.

\subsection{Sensitivity and limiting behavior of $\sigma$}

\begin{proposition}[Sensitivity formula of $\sigma$]\label{prop:sensitivity}
Let $\boldsymbol{\Gamma}(\sigma)$ be the DomiRank vector given by equation \eqref{eq:solution}. Then
\begin{align}\label{eq:sensitivity}
\frac{d\boldsymbol{\Gamma}}{d\sigma} = \theta\left(\sigma A + I\right)^{-2} A \mathbf{1}_{N \times 1}.
\end{align}
\end{proposition}

\begin{proof}
Write $M(\sigma) = \sigma A + I$. From $\boldsymbol{\Gamma} = \theta\sigma M^{-1} A\mathbf{1}$ and the inverse-matrix derivative formula $\dfrac{d M^{-1}}{d\sigma} = -M^{-1} A M^{-1}$,
\[
\begin{split}
\frac{d\boldsymbol{\Gamma}}{d\sigma}
&= \theta M^{-1} A\mathbf{1} - \theta\sigma M^{-1} A M^{-1} A\mathbf{1} \\
&= \theta M^{-1}\left(I - \sigma A M^{-1}\right) A\mathbf{1}
 = \theta M^{-1}\left(M - \sigma A\right) M^{-1} A\mathbf{1}.
\end{split}
\]
Since $M - \sigma A = I$, we obtain
\[
\dfrac{d\boldsymbol{\Gamma}}{d\sigma} = \theta M^{-2} A\mathbf{1} = \theta\left(\sigma A + I\right)^{-2} A\mathbf{1}.
\]
\end{proof}

Formula \eqref{eq:sensitivity} shows that the sensitivity of DomiRank with respect to $\sigma$ is given by the ``second-order suppression operator'' $(\sigma A + I)^{-2}$ acting on the degree vector: $\sigma$ not only changes the DomiRank values, but also reshapes the relative ranking of nodes in a graph-structure-dependent way. A natural question is whether there exist graphs for which $\sigma$ is ``ineffective'' (i.e., the ranking does not change with $\sigma$). The following proposition gives a complete characterization.

\begin{proposition}[Characterization of $\sigma$-invariance]\label{prop:invariance}
Let $G$ be a connected graph on $N \ge 2$ vertices and take $\theta = 1$. (The case $N = 1$ is excluded because $\boldsymbol{\Gamma} \equiv \boldsymbol{0}$ and the normalized distribution is then undefined.) Then the normalized DomiRank distribution is independent of $\sigma$ on $\left(0, -\dfrac{1}{\lambda_N}\right)$ if and only if the degree vector $\boldsymbol{k} = A\mathbf{1}_{N \times 1}$ is an eigenvector of the adjacency matrix $A$ (in which case the corresponding eigenvalue must be the largest eigenvalue $\lambda_1$, i.e., $G$ is a harmonic graph). In particular, the normalized DomiRank distribution, the degree-centrality distribution, and the principal-eigenvector-centrality distribution then coincide.
\end{proposition}

\begin{proof}
\noindent\textbf{Sufficiency:} Suppose $A\boldsymbol{k} = \lambda \boldsymbol{k}$. Since $G$ is connected, every component of $\boldsymbol{k}$ (a vertex degree) is a positive integer; by the Perron--Frobenius theorem, a positive eigenvector must correspond to $\lambda_1$, so $\lambda = \lambda_1$. Then
\[
\left(\sigma A + I\right)\boldsymbol{k} = \left(\sigma\lambda_1 + 1\right)\boldsymbol{k}
\quad\Longrightarrow\quad
\left(\sigma A + I\right)^{-1}\boldsymbol{k} = \frac{1}{\sigma\lambda_1 + 1}\boldsymbol{k},
\]
and hence
\[
\boldsymbol{\Gamma} = \theta\sigma\left(\sigma A + I\right)^{-1} A\mathbf{1} = \frac{\theta\sigma}{\sigma\lambda_1 + 1}\boldsymbol{k} \propto \boldsymbol{k}.
\]
The normalized distribution is always the degree distribution, independent of $\sigma$. Since $\boldsymbol{k}$ is the Perron vector, the three distributions coincide.

\noindent\textbf{Necessity:} Suppose the normalized DomiRank distribution is the constant vector $\boldsymbol{p}$, i.e., $\boldsymbol{\Gamma}(\sigma) = c(\sigma)\,\boldsymbol{p}$. By the Neumann-series expansion in the proof of Theorem~\ref{thm:sigma_zero}, as $\sigma \to 0$, $\boldsymbol{\Gamma} = \theta\sigma\boldsymbol{k} + O(\sigma^2)$; hence $\boldsymbol{\Gamma}/\|\boldsymbol{\Gamma}\| \to \boldsymbol{k}/\|\boldsymbol{k}\|$, i.e., $\boldsymbol{p} \propto \boldsymbol{k}$. Thus $(\sigma A + I)^{-1}\boldsymbol{k} \propto \boldsymbol{k}$ for every $\sigma$, meaning $\boldsymbol{k}$ is an eigenvector of $(\sigma A + I)^{-1}$ and hence of $A$.
\end{proof}

\begin{remark}
Lemma~\ref{lemma:regular} shows that the normalized DomiRank is uniform on regular graphs; Proposition~\ref{prop:invariance} gives a finer characterization: the exact graph class on which $\sigma$-invariance holds is the class of harmonic graphs, i.e., graphs whose degree vector is an eigenvector of the adjacency matrix (the classical notion of harmonic, or pseudo-regular, graphs \cite{p31,p32}). Regular graphs form a proper subclass: e.g., the 7-vertex harmonic tree with degree sequence $(3,2,2,2,1,1,1)$ satisfies $A\boldsymbol{k} = 2\boldsymbol{k}$ but is not regular. By contrast, biregular bipartite graphs are harmonic if and only if they are regular: for a biregular bipartite graph whose bipartition degrees are $a$ and $b$, $A\boldsymbol{k} = ab\,\mathbf{1}$, which is parallel to $\boldsymbol{k}$ only when $a = b$. In other words, the specificity of $\sigma$ manifests precisely on non-harmonic graphs.
\end{remark}

\begin{proposition}[Limiting behavior as $\sigma \to \sigma^*$]\label{prop:limit}
Let $G$ be a connected graph whose smallest eigenvalue $\lambda_N$ is simple, with corresponding unit eigenvector $\boldsymbol{v}_N$, and suppose $c_N := \boldsymbol{v}_N^T A\mathbf{1}_{N \times 1} \neq 0$. Then as $\sigma \to \sigma^* := -\dfrac{1}{\lambda_N}$,
\begin{equation}\label{eq:limit}
\frac{\boldsymbol{\Gamma}(\sigma)}{\|\boldsymbol{\Gamma}(\sigma)\|_2} \longrightarrow \mathrm{sgn}(c_N)\,\boldsymbol{v}_N .
\end{equation}
i.e., under high competition intensity, DomiRank converges to the direction of the\linebreak[4] smallest-eigenvalue eigenvector.
\end{proposition}

\begin{proof}
Let the spectral decomposition of $A$ be $A = \sum_{i=1}^{N} \lambda_i \boldsymbol{v}_i \boldsymbol{v}_i^T$ and write $c_i = \boldsymbol{v}_i^T A\mathbf{1}$. Then
\[
\boldsymbol{\Gamma}(\sigma) = \theta\sigma\left(\sigma A + I\right)^{-1} A\mathbf{1} = \theta\sigma\sum_{i=1}^{N} \frac{c_i}{1 + \sigma\lambda_i}\boldsymbol{v}_i.
\]
As $\sigma \to \sigma^*$, $1 + \sigma\lambda_N \to 0^+$; since $\lambda_N$ is simple, all other denominators $1 + \sigma^*\lambda_i > 0$ have a positive lower bound. Hence the $N$-th term dominates:
\[
\boldsymbol{\Gamma}(\sigma) = \frac{\theta\sigma^*\, c_N}{1 + \sigma\lambda_N}\,\boldsymbol{v}_N + O(1),
\qquad
\|\boldsymbol{\Gamma}(\sigma)\|_2 = \frac{\theta\sigma^*\, |c_N|}{1 + \sigma\lambda_N}\left(1 + o(1)\right).
\]
Dividing the two expressions yields \eqref{eq:limit}.
\end{proof}

\begin{remark}[Continuous interpolation mechanism of $\sigma$]\label{rem:interpolation}
Combining Theorem~\ref{thm:sigma_zero} with Proposition~\ref{prop:limit}: as $\sigma \to 0$, DomiRank degenerates to degree centrality (purely local information), and as $\sigma \to \sigma^*$, it converges to smallest-eigenvector centrality (global negative-spectral structure, encoding the approximate bipartivity of the network). The competition parameter $\sigma$ makes DomiRank interpolate continuously between these two extremes. We note that DomiRank is a member of the classical family of resolvent (Bonacich-type) centralities: $\boldsymbol{\Gamma} = \theta\sigma(\sigma A + I)^{-1}\boldsymbol{k} = c(I - \alpha A)^{-1}\boldsymbol{\beta}$ with $\alpha = -\sigma$, $\boldsymbol{\beta} = \boldsymbol{k}$, and $c = \theta\sigma$ (cf.\ Katz \cite{p17} and Bonacich \cite{p15,p16}); moreover, PageRank itself contains a tunable damping factor. The interpolation behaviour of such parameter-dependent resolvent measures is not new: Benzi and Klymko \cite{p18} proved that Katz-type and subgraph centralities interpolate between degree and eigenvector centrality as the algebraic parameter varies, and characterized the endpoints of that interpolation. We therefore do not claim the interpolation phenomenon as a new discovery; the specificity of DomiRank within this family lies on three axes. (i) \emph{Endpoint.} The limit \eqref{eq:limit} is the \emph{smallest} eigenvector, not the principal one. This requires the negative-parameter regime $\alpha = -\sigma < 0$, which produces non-positive values, whereas the standard positive-parameter Katz/Bonacich range studied in \cite{p18} cannot reach a signed, negative-spectrum eigenvector; the two endpoints thus lie in different qualitative regimes, not merely at different parameter values. (ii) \emph{Spectral datum and interpretation.} In \cite{p18} the family is parameterized over generic non-negative algebraic parameters with $\boldsymbol{\beta} = \boldsymbol{1}$; here the choice $\boldsymbol{\beta} = \boldsymbol{k}$ (the degree vector) carries the competition interpretation of $\sigma$, so that the resolvent acts on a degree-indexed datum and the competition threshold $\theta$ remains visible in the closed form. (iii) \emph{Degeneracy.} \cite{p18} concerns generic members of the family on graphs admitting the relevant limits; here the family provably degenerates on an identifiable class, namely the harmonic graphs, i.e.\ graphs whose degree vector is an eigenvector of $A$ (Proposition~\ref{prop:invariance}), on which $\sigma$ is entirely inert and the entropy is constant (Lemma~\ref{lemma:polarization}(i)). This exact characterization of where the parameter becomes vacuous is, to our knowledge, not part of the existing theory of resolvent centralities.
\end{remark}

\begin{remark}[Relation to the literature on bipartivity and negative eigenvectors]\label{rem:bipartivity}
The limit object $\boldsymbol{v}_N$ of Proposition~\ref{prop:limit} is not new: the sign structure of the smallest eigenvalue and its eigenvector is a classical tool for characterizing the approximate bipartivity of networks --- Estrada and Rodr\'iguez-Vel\'azquez \cite{p20} proposed bipartivity measures based on the negative spectrum that quantify how far a network is from being bipartite; in social-network analysis with negative relations, Bonacich and Lloyd \cite{p21} used the smallest eigenvector to define a signed status index; the original DomiRank paper \cite{p11} also noted that the solution is dominated by the smallest eigenvector as $\sigma \to -1/\lambda_N$. Therefore, the contribution of DomiRank does not lie in the limiting endpoints themselves, but in the fact that $\sigma$ provides a continuous path connecting degree centrality (local, positive) with smallest-eigenvector centrality (global, signed, carrying bipartivity information): it is the interpolation mechanism, rather than either endpoint, that distinguishes DomiRank from the static spectral indices of \cite{p20,p21}.
\end{remark}

\begin{remark}[Ranking reversal]\label{rem:reversal}
On non-harmonic graphs, the node ranking can reverse as $\sigma$ increases: a node ranked high at low $\sigma$ (determined by the local degree) may be overtaken at high $\sigma$ by a node whose neighbor structure is more dominant. Take the Karate club network as an example (see Table~\ref{tab:networks}): node 3 (degree 10) ranks above node 2 (degree 9) at low $\sigma$, but the two cross at $\sigma/\sigma^* \approx 0.390$; at $\sigma/\sigma^* = 0.9$, $\Gamma_2 = 2.729$ is far above $\Gamma_3 = 1.446$. Figure~\ref{figure11} shows the evolution trajectories of all nodes' DomiRank values with $\sigma$.
\end{remark}

\begin{figure}[!h]
	\centering
	\includegraphics[width=0.75\textwidth]{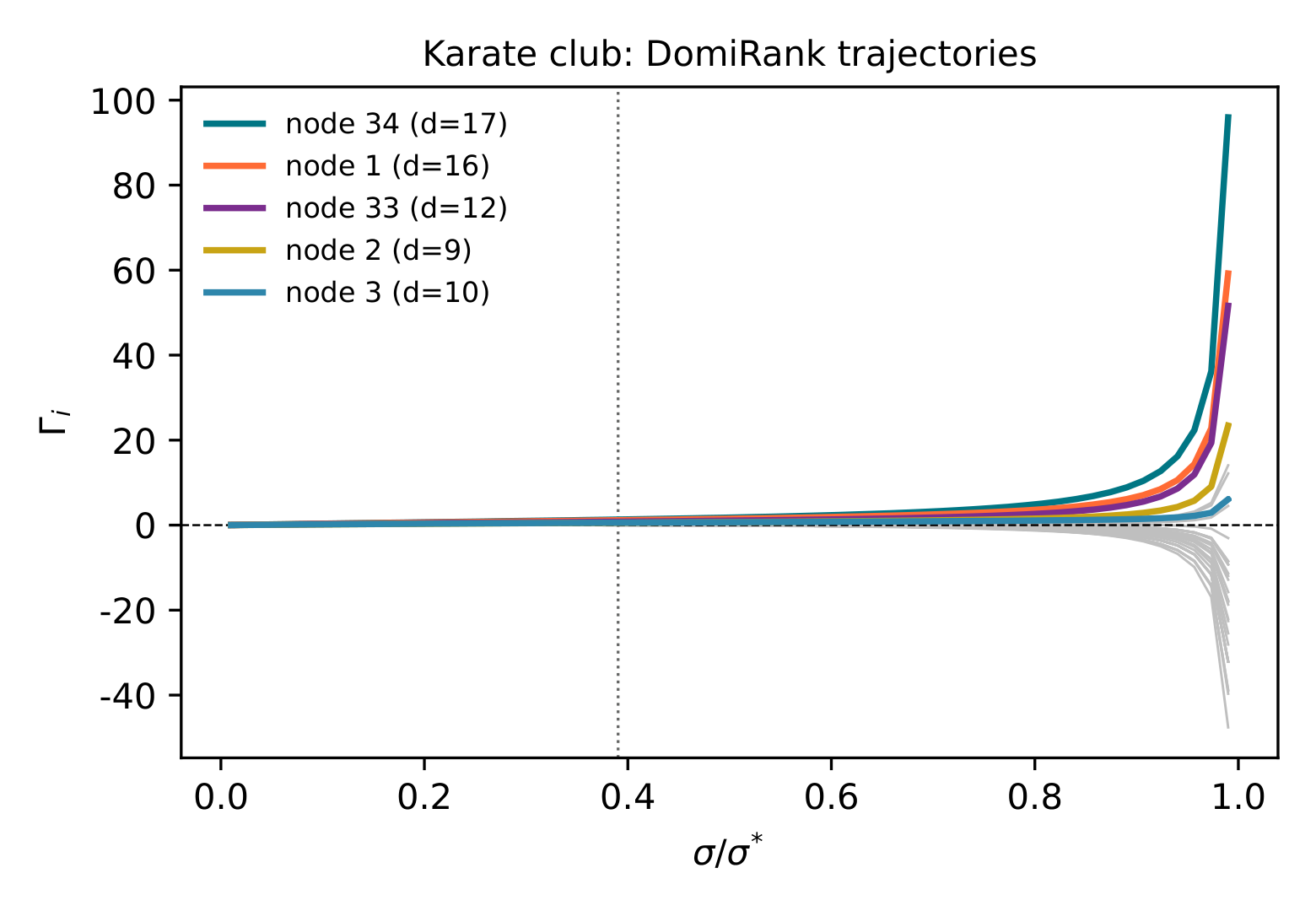}\\
	\caption{Evolution of the DomiRank values of all nodes with $\sigma/\sigma^*$ in the Karate club network.}
	\textbf{Alt text:}\ A family of curves showing the DomiRank value of every vertex of the Karate club network against the relative competition intensity. The trajectory of node 34 rises steeply while those of the other nodes stay low, and the curves of nodes 2 and 3 cross at an intensity near 0.39.
	\label{figure11}
\end{figure}

\subsection{Real-world network data and\\allowbreak experimental setup}

To verify the above theoretical results, we select four classical real-world networks from different domains: the Zachary karate club social network (Karate) \cite{p33}, the bottlenose dolphin social network (Dolphins) \cite{p34}, the co-occurrence network of characters in {\it Les Mis\'erables} (LesMis) \cite{p35}, and the US political books copurchase network (Polbooks) \cite{p36}. All networks are restricted to their largest connected component and treated as unweighted undirected graphs; their statistics and spectral parameters are in Table~\ref{tab:networks}. For each network we take 60 equidistant sampling points on the relative grid $\sigma = x\,\sigma^*$, $x$ equidistant in $[0.01,\; 0.99]$, and solve $(\sigma A + I)\boldsymbol{\Gamma} = \sigma A\mathbf{1}$ directly for the DomiRank vector; ranking correlations use Kendall's $\tau$ \cite{p37} in the tie-aware (variant-$b$) convention, which is the one that matches the comparison of $\sigma=0$ against the degree ranking, where every tie is genuine. Breaking ties instead by an arbitrary positional rank (for instance via \texttt{argsort}) gives noticeably smaller values, on the order of $0.82$ rather than $0.93$ for $\tau(\boldsymbol{\Gamma},\text{degree})$ at the low-$\sigma$ end; the convention is stated here to make the tables reproducible.

\begin{table}[!h]
\centering
\caption{Basic statistics and spectral parameters of the four real-world networks}
\begin{tabular}{lcccccc}
\hline
Network & $N$ & $M$ & $\langle k \rangle$ & $\lambda_1$ & $\lambda_N$ & $\sigma^* = -1/\lambda_N$ \\
\hline
Karate   & 34  & 78  & 4.59 & 6.726  & $-4.487$ & 0.2229 \\
Dolphins & 62  & 159 & 5.13 & 7.194  & $-3.821$ & 0.2617 \\
LesMis   & 77  & 254 & 6.60 & 12.006 & $-5.267$ & 0.1899 \\
Polbooks & 105 & 441 & 8.40 & 11.933 & $-5.201$ & 0.1923 \\
\hline
\end{tabular}
\label{tab:networks}
\end{table}

\subsection{Experimental results and analysis}

\textbf{(1) Verification of entropy monotonicity.} The theoretical entropy \eqref{eq:entropy} is defined for positive entries only, whereas DomiRank admits non-positive values under strong competition (see item (3) below). We therefore compute the experimental entropy on the \emph{effective distribution} obtained by truncating negative DomiRank values to zero and renormalizing (Algorithm 1); in the range where all entries are positive this coincides with the theoretical entropy. Figure~\ref{figure12} shows the normalized effective entropy $H(\sigma)/\log_2 N$ versus $\sigma/\sigma^*$ on the four networks. On all networks the entropy decreases monotonically with $\sigma$, consistent with Theorem~\ref{thm:entropy_decrease} and Conjecture~\ref{conj:entropy}; for example, on Karate it decreases from $0.9188$ ($\sigma/\sigma^* = 0.05$) to $0.4989$ ($\sigma/\sigma^* = 0.9$), showing that increasing the competition intensity continuously concentrates the importance distribution toward a few nodes.

\begin{figure}[!h]
	\centering
	\includegraphics[width=0.75\textwidth]{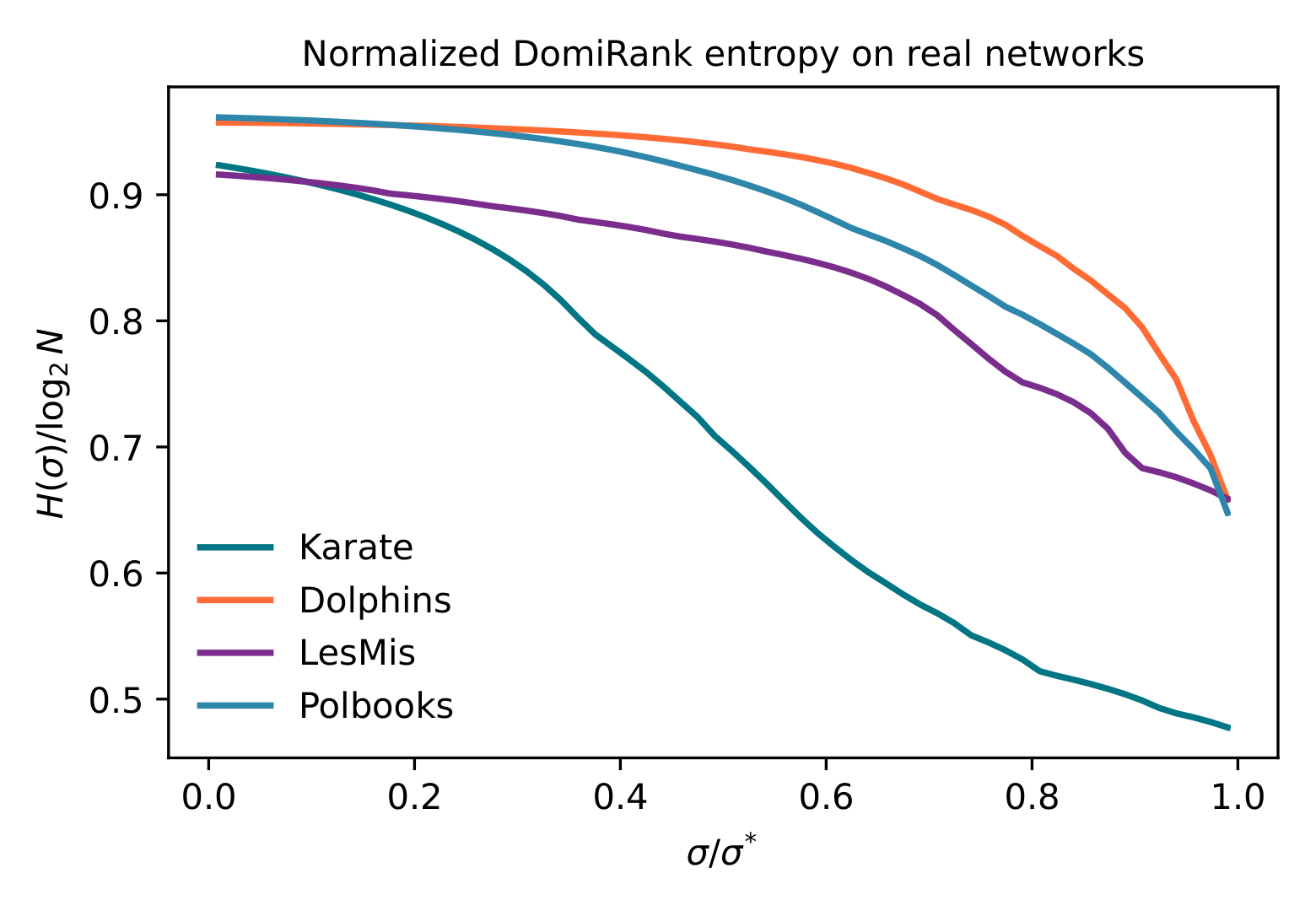}\\
	\caption{Normalized effective DomiRank entropy $H(\sigma)/\log_2 N$ versus $\sigma/\sigma^*$ on the four real-world networks (computed on the effective distribution of Algorithm 1). On the initial segment where all DomiRank entries are positive the effective distribution coincides with the theoretical one and the curves depict the theoretical entropy; the right-hand tails, where negative entries arise, depict the truncated effective entropy.}
	\textbf{Alt text:}\ A line plot of the DomiRank entropy of the four real-world networks against the relative competition intensity. All four curves decrease monotonically across the grid.
	\label{figure12}
\end{figure}

\textbf{(2) Verification of continuous ranking interpolation.} Figure~\ref{figure13} shows the Kendall $\tau$ correlation between the DomiRank ranking and four reference rankings versus $\sigma/\sigma^*$. At the low-$\sigma$ end, $\tau(\boldsymbol{\Gamma}, \text{degree})$ ranges between $0.9264$ and $0.9551$, consistent with (but not equal to $1$, since the sampling floor is $x = 0.01$ rather than $\sigma = 0$, and second-order terms are still present) the degree-centrality limit of Theorem~\ref{thm:sigma_zero}; as $\sigma$ increases, this correlation decays continuously, dropping to $0.1979$--$0.3335$ at $\sigma/\sigma^* = 0.9$. At the same time, $\tau(\boldsymbol{\Gamma}, \mathrm{sgn}(c_N)\boldsymbol{v}_N)$ increases monotonically and reaches $0.9945$ (Karate), $0.8675$ (Dolphins), $0.9720$ (LesMis), and $0.7531$ (Polbooks) at the endpoint of the sampling range ($\sigma/\sigma^* = 0.99$), verifying the smallest-eigenvector limit of Proposition~\ref{prop:limit}. Notably, the correlation between DomiRank and principal-eigenvector centrality decays rapidly toward zero in the medium-to-high $\sigma$ region (e.g., $\tau \approx -0.0764$ on Karate at $\sigma/\sigma^* = 0.9$), showing that DomiRank is not a simple deformation of eigenvector centrality but characterizes a genuinely different dominance structure.

\begin{figure}[!h]
	\centering
	\includegraphics[width=0.95\textwidth]{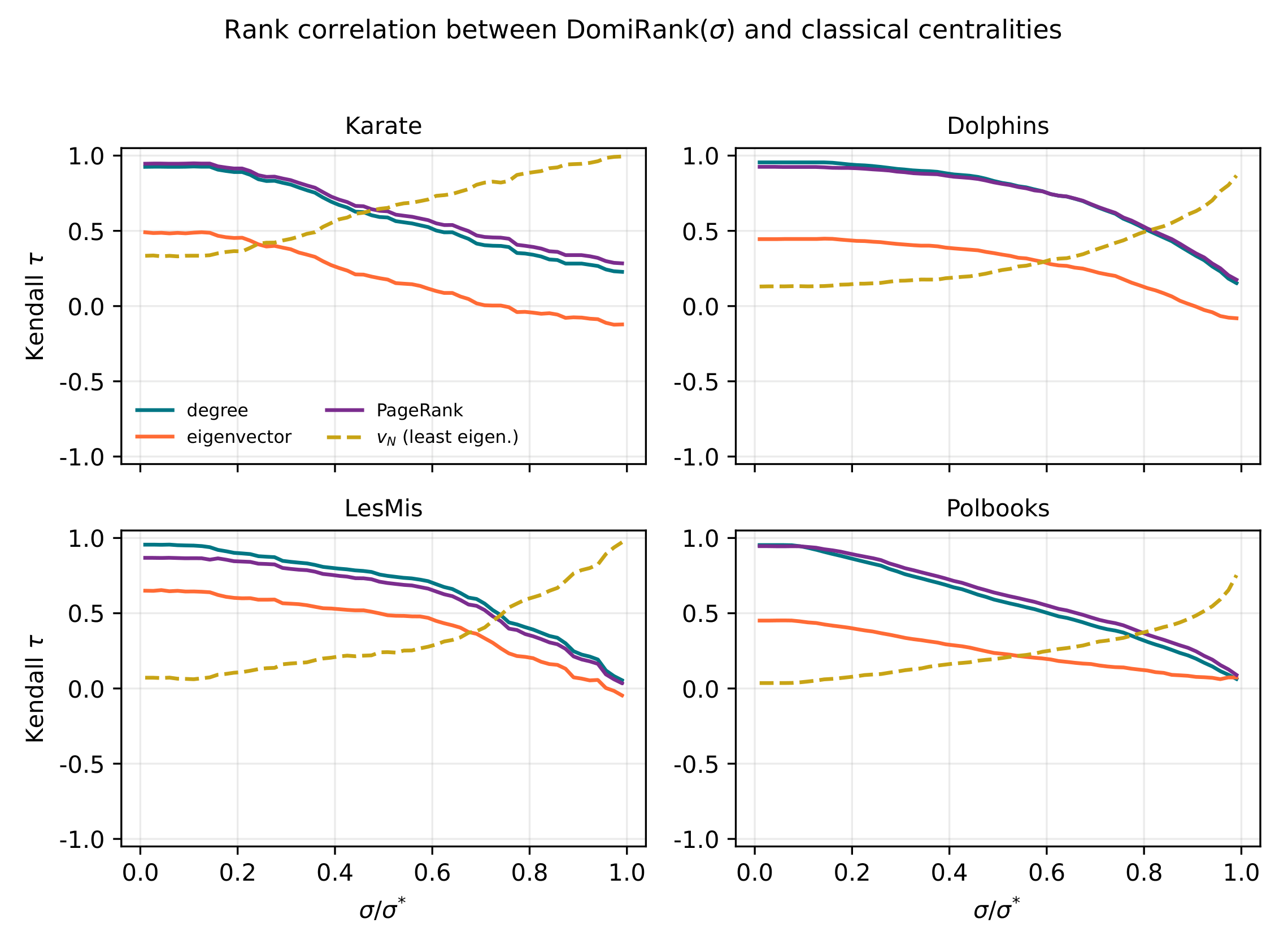}\\
	\caption{Kendall $\tau$ between the DomiRank$(\sigma)$ ranking and the degree, principal-eigenvector, PageRank, and smallest-eigenvector rankings versus $\sigma/\sigma^*$.}
	\textbf{Alt text:}\ A line plot of Kendall's tau between the DomiRank ranking and four reference rankings, against the relative competition intensity. The correlation with the degree ranking falls from about 0.93 to between 0.20 and 0.33, while the correlation with the smallest-eigenvector ranking rises to between 0.75 and 0.99.
	\label{figure13}
\end{figure}

\textbf{(3) Positive-fraction decay and the emergence of the dominance core.} DomiRank allows nodes to take non-positive values --- a distinctive phenomenon separating it from degree, eigenvector, and PageRank centralities, which are always positive; a negative node can be interpreted as having its dominance completely suppressed by its neighbors. Figure~\ref{figure14} shows that the fraction of positive nodes decays continuously from $100\%$ at the low-$\sigma$ end, reaching $23.5\%$ (Karate), $48.4\%$ (Dolphins), $51.9\%$ (LesMis), and $40.0\%$ (Polbooks) at the endpoint of the sampling range. That is, as competition intensifies, only the ``dominance core'' of the network remains, providing real-network evidence for the theoretical discussion in Remark~\ref{rem:pendant_sharp} that DomiRank values can exceed the dominance threshold $\theta$.

\begin{figure}[!h]
	\centering
	\includegraphics[width=0.75\textwidth]{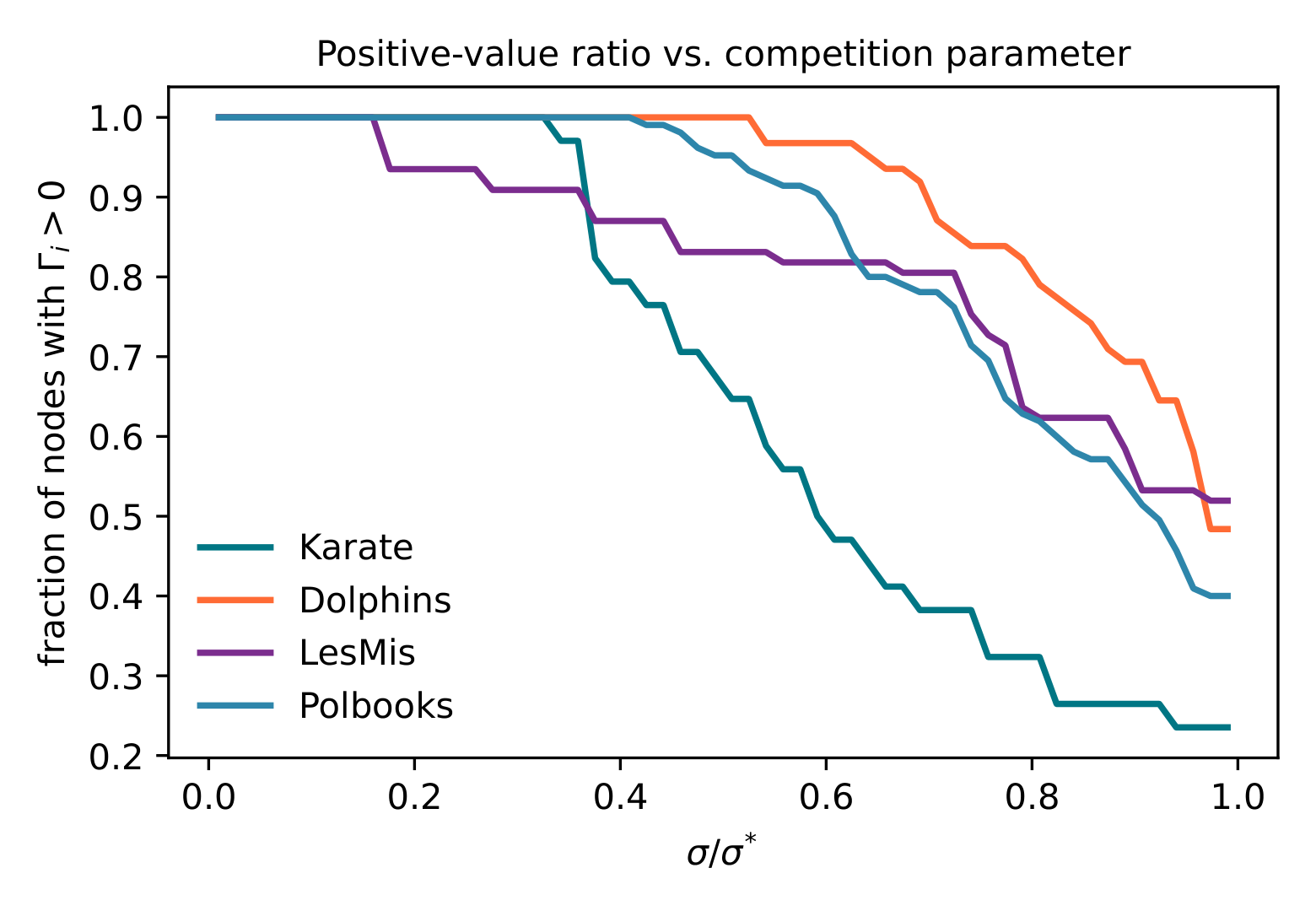}\\
	\caption{Fraction of nodes with positive DomiRank values versus $\sigma/\sigma^*$ on the four real-world networks.}
	\textbf{Alt text:}\ A line plot of the fraction of vertices with strictly positive DomiRank values against the relative competition intensity, for the four real networks. All four curves fall from 100 percent to between 23 and 52 percent.
	\label{figure14}
\end{figure}

\textbf{(4) Evolution of the key-node set.} Table~\ref{tab:top5} lists the top-5 key nodes of the Karate network under low, medium, and high competition intensities, compared with degree centrality, eigenvector centrality, and PageRank. At low $\sigma$, DomiRank coincides completely with degree centrality; at medium and high $\sigma$, node 2 overtakes node 3 (ranking reversal), and the magnitude of the DomiRank values expands sharply (the $\Gamma$ value of node 34 grows from $0.1583$ to $10.4660$), exponentially widening the gap between the head nodes and the rest, strongly reinforcing the dominant key nodes.

\begin{table}[!h]
\centering
\caption{Top-5 key nodes of the Karate network under different competition intensities}
\begin{tabular}{ll}
\hline
Method / parameter & Top-5 nodes (descending) \\
\hline
DomiRank ($\sigma/\sigma^* \approx 0.04$) & 34, 1, 33, 3, 2 \\
DomiRank ($\sigma/\sigma^* \approx 0.51$) & 34, 1, 33, 2, 3 \\
DomiRank ($\sigma/\sigma^* \approx 0.91$) & 34, 1, 33, 2, 3 \\
Degree centrality & 34, 1, 33, 3, 2 \\
Eigenvector centrality & 34, 1, 3, 33, 2 \\
PageRank & 34, 1, 33, 3, 2 \\
\hline
\end{tabular}
\label{tab:top5}
\end{table}

In summary, the real-world experiments show that $\sigma$ is not a redundant parameter: it drives the DomiRank ranking on real networks to evolve systematically from the degree ranking to the smallest-eigenvector ranking, accompanied by a series of quantifiable specificity phenomena --- entropy monotonic decrease, positive-fraction decay, and ranking reversals. In practice, $\sigma$ can be chosen according to the application goal: weak competition (small $\sigma$) corresponds to ``important-node'' identification, while strong competition (large $\sigma$) corresponds to ``dominant key-node'' identification.

\subsection{Large-scale synthetic networks: verifying\\allowbreak Propositions~\ref{prop:initial_entropy} and~\ref{prop:limit} beyond the classical benchmarks}
\label{sec:large}

The four networks of Table~\ref{tab:networks} are small and hand-picked, so a natural objection is that the phenomena above are artefacts of a particular choice of benchmarks. We therefore repeat the $\sigma$-scan on synthetic network families of two sizes, $N \approx 10^3$ and $N \approx 10^4$, in the three standard generators Erd\H{o}s--R\'enyi (ER, mean degree $6$), Barab\'asi--Albert (BA, $m = 5$) and Watts--Strogatz (WS, $k = 6$, $\beta = 0.1$), each reduced to its largest connected component. This subsection verifies two claims: that the entropy behaves as described in Section~\ref{sec:sigma} at scales three orders of magnitude beyond the classical benchmarks, and that the interpolation endpoints of Theorem~\ref{thm:sigma_zero} and Proposition~\ref{prop:limit} are a property of the model rather than of the four networks. We also report the harmonicity residual $\|A\boldsymbol{k} - \bar{d}\,\boldsymbol{k}\|_\infty / \|\boldsymbol{k}\|_\infty$, where $\boldsymbol{k}$ is the degree vector and $\bar{d}$ the mean degree, which vanishes exactly on the harmonic graphs of Proposition~\ref{prop:invariance}; it is large throughout, so none of these graphs is harmonic and $\sigma$ is genuinely active on all of them.

\begin{table}[!h]
\centering
\caption{The $\sigma$-scan on synthetic network families. $H$ is the DomiRank entropy at the two ends of the relative grid $x = \sigma/\sigma^* \in [0.01, 0.99]$ and $\Delta H$ the net change; ``mono.'' records whether $H$ decreases at every grid step and ``frac.'' the fraction of steps at which it does. $\tau_{\mathrm{deg}}$ and $\tau_{N}$ are the tie-aware Kendall correlations of the DomiRank ranking with the degree and smallest-eigenvector rankings at $x = 0.01$ and $x = 0.99$. The harmonicity residual is $\bar{d}$-normalized and vanishes only on harmonic graphs. Grid sizes differ because sparse LU factorization of $\sigma A + I$ dominates the cost at $N \approx 10^4$; every grid point is an exact linear solve.}
\label{tab:large}
{\scriptsize
\setlength{\tabcolsep}{3.5pt}
\begin{tabular}{llrrlrrlrr}
\hline
Model & $N$ & grid & $\lambda_N$ & $H(x{=}0.01)$ & $H(x{=}0.99)$ & $\Delta H$ & mono. & frac. \\
\hline
ER  & 1000  & 60 & $-5.7445$ & 9.8566 & 9.3858 & $-0.4708$ & yes & 1.0000 \\
BA  & 1000  & 60 & $-10.9131$ & 9.5547 & 8.3679 & $-1.1868$ & yes & 1.0000 \\
WS  & 998   & 60 & $-3.2669$ & 9.9183 & 9.5894 & $-0.3289$ & yes & 1.0000 \\
ER  & 9994  & 20 & $-5.7452$ & 13.1812 & 12.7019 & $-0.4793$ & yes & 1.0000 \\
BA  & 10000 & 20 & $-19.0779$ & 12.8141 & 11.9178 & $-0.8963$ & yes & 1.0000 \\
WS  & 9993  & 20 & $-3.4066$ & 13.2447 & 13.0085 & $-0.2362$ & yes & 1.0000 \\
\hline
\end{tabular}
}
\vspace{4pt}
{\scriptsize
\setlength{\tabcolsep}{3.5pt}
\begin{tabular}{llrrrrrrr}
\hline
 & & \multicolumn{2}{c}{pos. fraction} & \multicolumn{2}{c}{$\tau_{\mathrm{deg}}$} & \multicolumn{2}{c}{$\tau_{N}$} & harm. \\
Model & $N$ & $x{=}0.01$ & $x{=}0.99$ & $x{=}0.01$ & $x{=}0.99$ & $x{=}0.01$ & $x{=}0.99$ & residual \\
\hline
ER  & 1000  & 1.0000 & 0.8750 & 0.9927 & 0.5207 & 0.0355 & 0.3037 & 2.04 \\
BA  & 1000  & 1.0000 & 0.6130 & 0.9858 & 0.1395 & $-0.0004$ & $-0.6124$ & 9.45 \\
WS  & 998   & 1.0000 & 0.9208 & 0.9531 & 0.6801 & $-0.0632$ & $-0.0757$ & 1.41 \\
ER  & 9994  & 1.0000 & 0.8851 & 0.9932 & 0.5361 & $-0.0155$ & $-0.4913$ & 2.87 \\
BA  & 10000 & 1.0000 & 0.7438 & 0.9844 & 0.2274 & 0.0337 & 0.4246 & 13.98 \\
WS  & 9993  & 1.0000 & 0.9593 & 0.9555 & 0.7441 & 0.0101 & 0.0027 & 1.73 \\
\hline
\end{tabular}
}
\end{table}

Three findings follow. First, the entropy decreases at \emph{every} grid step on all six graphs, with the net decrease ranging from $-0.24$ (WS, $N \approx 10^4$) to $-1.19$ (BA, $N = 10^3$); this is consistent with Conjecture~\ref{conj:entropy} far beyond the four benchmarks, and the decrease is not a small-$\sigma$ effect, continuing past $x = 0.5$, well outside the range of the sufficient conditions of Theorem~\ref{thm:entropy_decrease}. Second, the low-competition endpoint is reproduced with high fidelity ($\tau_{\mathrm{deg}} \in [0.953, 0.993]$ at $x = 0.01$), after which the correlation with the degree ranking decays to $0.14$--$0.74$, most sharply on the strongly hub-dominated BA graphs. Third, and against our initial expectation, $\tau_N$ at $x = 0.99$ does not approach $1$ uniformly: it is $0.30$ on ER ($N = 10^3$), $-0.61$ on BA ($N = 10^3$), $-0.49$ on ER ($N = 10^4$) and $0.42$ on BA ($N = 10^4$), while remaining near zero on WS. The reason is visible in the same table: the fraction of strictly positive entries at $x = 0.99$ falls to $0.61$--$0.96$, so the ranking being correlated is no longer a probability distribution and the comparison with $\boldsymbol{v}_N$ loses part of its meaning; the limit $\sigma \to \sigma^*$ is a statement about the unnormalized solution vector (Proposition~\ref{prop:limit}), not a guarantee of high rank correlation at the last grid point. We report this rather than tuning the grid, because it is exactly the regime in which DomiRank is used to discriminate dominant nodes, and the negative values are the signature of the mechanism rather than a defect of the measure.

\subsection{A running example: the 5-vertex graph of Example~\ref{ex:first} through the whole scan}\label{sec:anchor}

The 5-vertex graph of Example~\ref{ex:first} is small enough to be followed exactly, and it is worth returning to once, since it puts the three parts of the paper on a single object. Its degree sequence is $(2,2,4,1,1)$ and $\lambda_N = -1.8136$, so $\sigma^* \approx 0.5514$. On the 60-point relative grid of Section~6.2, all five entries stay strictly positive up to $x = 0.5083$ ($\sigma = 0.2803$), and on that domain the theoretical entropy of Proposition~\ref{prop:initial_entropy} decreases monotonically from $H = 2.1179$ at $x = 0.01$ through $2.0360$ at $x = 0.16$ to $1.2033$ at $x = 0.5083$, exactly as Conjecture~\ref{conj:entropy} requires: the edge covariance here is $\mathrm{Cov}(k_i, \log_2 k_j) = -0.72 < 0$, the disassortative side. The positive fraction then leaves $1$ at $x = 0.5249$ (three of five vertices positive), and the effective entropy of the truncated distribution, the object plotted in Figure~\ref{figure12}, reaches $H = 0$ at $x = 0.7575$ and stays there, a single vertex carrying the whole positive mass. Over the full grid the effective entropy decreases on $45$ of the $59$ steps and is constant on the other $14$, so it is not strictly decreasing throughout. That is a property of the truncation rather than a counterexample: the step at which the entropy stops moving is exactly the step at which one vertex is left, and the theoretical entropy is monotone throughout the region where it is defined. The same graph satisfies the symmetry axiom in the strictest way, the two orbits $\{1,2\}$ and $\{4,5\}$ receiving identical values at every $\sigma$, as Example~\ref{ex:first} shows at $\sigma = 0.25$.

Finally, we returned to the entropy-minimality claim of Conjecture~\ref{conj:star} on a random sample rather than an exhaustive one: 722 random graphs (ER, BA and WS with $N \in \{10, 20, 50\}$, reduced to the largest connected component) were compared with the star $S_N$ at the 19 relative intensities $x \in \{0.05, 0.10, \dots, 0.95\}$ of the operationalization given after Conjecture~\ref{conj:star}. Across $13{,}718$ comparisons, $H_G \ge H_{S_N}$ held in every case, with no counterexample; this complements, and does not replace, the exhaustive verification over all connected graphs with $N \le 7$ reported in Remark~\ref{rem:star_evidence}, which used a separate batch of 900 random networks at the same three values of $N$.

\section{Relation to companion work}
\label{sec:relation}

This paper characterizes DomiRank as a mathematical object: its behaviour under graph automorphisms, its entropy, its response to graph transformations, and the spectral meaning of $\sigma$. A companion study (Y.~Y.~Zhang and C.~Y.~Zhao, unpublished results) addresses a different question, namely whether a centrality measure's advantage over a degree-preserving (configuration-model) null model can be attributed to the measure itself. That companion work decomposes any null-model $z$-score into a fragility component, shared with degree centrality, and a specificity component, the increment over the degree baseline $z_{\mathrm{metric}} = z_{\deg} + \Delta z_{\mathrm{metric}}$, and applies the decomposition to DomiRank on the network-dismantling and influence-maximization tasks. It finds that on dismantling the increment $\Delta z_{\mathrm{DomiRank}}$ is not statistically distinguishable from zero across eighteen networks, while on influence maximization a positive increment reappears for several networks.

The two studies use the term \emph{specificity} in different senses, and we state the distinction to avoid a misreading. In the companion study specificity is a \emph{statistical} quantity, the part of a task-level advantage that survives subtraction of the degree baseline. Here the specificity of $\sigma$ is a \emph{structural} property of the model: while the parameter is inert on the class of harmonic graphs, where the normalized DomiRank distribution is invariant under $\sigma$ (Proposition~\ref{prop:invariance}), it actively reshapes the ranking off that class, interpolating continuously between degree centrality and smallest-eigenvector centrality (Theorem~\ref{thm:sigma_zero} and Proposition~\ref{prop:limit}). The two statements are not in tension: the companion study quantifies how much of DomiRank's \emph{task-level} edge over degree survives an explicit degree baseline, whereas this paper shows that $\sigma$ nevertheless has a well-defined effect on the ranking itself. Because the resolution of $\sigma$ is not uniform across graphs, the companion study's protocol of maximizing $z$ over a grid of $\sigma$ values is, if anything, favourable to DomiRank: it grants the measure a degree of freedom that the parameter-free baselines do not possess.

Two consequences for practice. First, null-model evaluations of centrality should report the degree baseline $z_{\deg}$ alongside $z_{\mathrm{metric}}$ and interpret only the increment $\Delta z_{\mathrm{metric}}$ as measure-specific evidence; the structural results derived here describe the measure, not its advantage over a degree baseline, and should not be read as a specificity claim on a downstream task. Second, an application seeking an advantage beyond degree centrality should expect it to be task-conditional, and should verify it on the task of interest rather than infer it from the parameterization alone.

\section{Conclusion}

This paper carried out a systematic study of the DomiRank algorithm. The main results are as follows:

\begin{enumerate}
    \item \textbf{Symmetry axiom and its combinatorial strength:} we proved that DomiRank is invariant under the automorphism group, and that the number of distinct DomiRank entries is bounded by the number of automorphism orbits; we gave the smallest-order example in which inequivalent vertices receive equal values, showing the bound to be strict. We also derived the DomiRank behaviour of special structures (regular and vertex-transitive graphs), proved that a connected graph has a constant DomiRank vector if and only if it is regular, and gave the constraints of local structures (identical neighborhoods, subset neighborhoods, and pendant vertices) on DomiRank values, with a counterexample showing that the threshold condition in the pendant-vertex lemma is indispensable.

    \item \textbf{Entropy bound and initial derivative:} we proved that the maximum DomiRank entropy $\log_2 N$ of a connected graph is attained only by regular graphs, and established, under sufficient conditions, the monotonic decrease of the entropy with $\sigma$; we further identified the initial entropy derivative with the degree--log-degree edge covariance (Proposition~\ref{prop:initial_entropy}), which settles the low-competition behavior rigorously on structurally disassortative graphs and yields an assortative counterexample showing that the unconditional monotonicity fails; the refined version is posed as Conjecture~\ref{conj:entropy}, supported by exhaustive and real-network numerics.

    \item \textbf{Spectral role of $\sigma$ and its degeneracy:} we derived the closed-form sensitivity formula $\theta(\sigma A+I)^{-2}A\mathbf{1}$; proved that the normalized DomiRank distribution is $\sigma$-invariant if and only if the degree vector is an eigenvector of the adjacency matrix (harmonic graphs), and that DomiRank converges to the smallest-eigenvector direction as $\sigma \to -1/\lambda_N$, so that $\sigma$ realizes a continuous interpolation between degree centrality and smallest-eigenvector centrality within the resolvent (Bonacich-type) family. Remark~\ref{rem:interpolation} states how this differs from the existing theory of parameter-dependent resolvent centralities: the endpoint is the smallest rather than the principal eigenvector, reachable only in the negative-parameter regime; the resolvent acts on the degree vector; and the degenerate class is characterized exactly. Experiments on four real networks and on six synthetic graphs with $10^3 \le N \le 10^4$ (Section~\ref{sec:large}) verified the $\sigma$-specific phenomena of entropy monotonic decrease, continuous ranking interpolation, positive-fraction decay, and ranking reversals, and showed that the entropy decreases at every grid step on all six synthetic graphs; they also delimit the high-$\sigma$ regime in which the positive-fraction decay makes the smallest-eigenvector limit a statement about the solution vector rather than about the induced ranking, and Section~\ref{sec:anchor} traces that distinction on a single 5-vertex graph, where the theoretical entropy stays monotone for as long as it is defined while the truncated entropy flattens at zero as soon as one vertex is left.

    \item \textbf{Response to structural transformations:} we studied the influence of vertex partitions, edge deletions, edge swaps, and $m$-products on DomiRank values and derived the corresponding theorems. In particular, we proved that the DomiRank values of the $m$-product graph $G^{(m)}$ are explicitly related to those of the original graph at parameter $\sigma m$.
\end{enumerate}

Compared with traditional algorithms such as principal-eigenvector centrality, DomiRank quantifies node dominance through a dynamical-system model and respects the symmetry constraints imposed by graph automorphisms, giving it distinctive evaluation behavior and providing theoretical tools and application support for node-importance evaluation in complex networks. The companion paper \cite{p22} further exploits the tunability of $\sigma$ established here in DomiRank-based fusion strategies for link prediction, where, under unified parameters, DomiRank-based fusion attains stable and statistically significant improvements over local baselines that PageRank-based fusion does not consistently provide. That comparison is carried out on the standard link-prediction benchmark; the implicit degree bias of that benchmark, and the degree-only reference level against which the fusion gain must therefore be read, are quantified in the companion paper itself.

Finally, we stress the scope of the present results. The theoretical statements collected here --- invariance under automorphisms, the entropy bound, the behaviour under graph transformations, and the spectral role of $\sigma$ --- characterize DomiRank as a \emph{measure}, not as a superior \emph{ranking method} on any particular task. Whether DomiRank's ranking carries information beyond what degree centrality already provides is a separate, task-level question, examined in a second companion study (Y.~Y.~Zhang and C.~Y.~Zhao, unpublished results) through an explicit degree-baseline decomposition of null-model $z$-scores; we refer the reader to Section~\ref{sec:relation} for a discussion of how the two notions of specificity relate.

\section*{Acknowledgements}

The authors thank the College of Sciences of China Jiliang University for institutional support.

\medskip

\noindent
\textbf{Declaration of generative AI and AI-assisted technologies in the
writing process.} The manuscript, including its text, figures, numerical analysis
and code, was written and produced by the authors. No generative-AI tool was used
to generate any part of it, and no such tool is listed as an author.

\section*{Funding}

This research received no specific grant from any funding agency in the public, commercial, or not-for-profit sectors.

\section*{Author contributions}

\textbf{Yingying Zhang:} Conceptualization, Methodology, Software, Formal analysis, Investigation, Data curation, Writing -- original draft, Visualization. \textbf{Chengye Zhao:} Conceptualization, Methodology, Supervision, Project administration, Writing -- review \& editing.

\section*{Competing interests}

The authors declare that they have no known competing financial interests or personal relationships that could have appeared to influence the work reported in this paper.

\section*{Data availability}

All figures and tables of Section~\ref{sec:sigma} are reproducible from the process described in Section~6.2: given the four edge lists, the DomiRank vector at each sampling point is obtained by the single linear solve $(\sigma A + I)\boldsymbol{\Gamma} = \sigma A\mathbf{1}$. The four networks are public benchmark datasets, taken from the sources given in Section~6.2 and used unchanged except for the restriction to the largest connected component; the exact edge lists used here, in the reduced form in which they were analysed, are deposited in the public archive accompanying this article (see Code availability below), together with the scripts and figures. No new data were collected for this study, and none are held back. The synthetic networks of Table~\ref{tab:large} and Section~\ref{sec:anchor} are generated by the procedure and parameters stated in Section~6.3.

\section*{Code availability}

The core computation underlying every DomiRank vector and every entropy value reported here is given in full as Algorithm~\ref{alg:grid}: one linear solve per sampling point, no iteration and no third-party solver. The remaining figures are analytical post-processing of those same vectors and are specified in the text where they appear: the Kendall $\tau$ correlations of Figure~\ref{figure13} use the tie-aware convention of Section~6.2, the positive fractions of Figure~\ref{figure14} are the counts of strictly positive entries of $\boldsymbol{\Gamma}$, and the trajectories of Figure~\ref{figure11} are the same vectors sampled along the grid. The full implementation used for the experiments of Section~\ref{sec:sigma}, together with the synthetic-network scans of Table~\ref{tab:large}, the four reduced edge lists and all fifteen figures are deposited in a public archive at \textbf{doi:10.5281/zenodo.23366403} \cite{p38}. The implementation depends only on dense or sparse linear algebra at each sampling point; the network edge lists are those of Zachary \cite{p33}, Lusseau et~al.\ \cite{p34}, Knuth \cite{p35} and Newman \cite{p36}. A manuscript on null-model tests of centrality specificity is cited above as unpublished results; it is neither under review nor part of the archive. Earlier versions of this work were posted to arXiv as arXiv:2610.00107 and have since been revised; the present version supersedes them.

\section*{Abbreviations}

\noindent BA, Barab\'asi--Albert; DomiRank, node dominance ranking;\\
ER, Erd\H{o}s--R\'enyi; KP, Katz--PageRank; LCC, largest connected component;\\
ORCID, Open Researcher and Contributor ID; PR, PageRank;\\
WS, Watts--Strogatz.

\end{document}